\documentclass[amsart,11pt,oneside,english]{article}
\usepackage{amssymb,amsmath,amsfonts, amsmath,mathtools,
amsthm,euscript,mathrsfs,MnSymbol,verbatim,enumerate,multirow,bbding,slashed,multicol,color,array, esint,babel, tikz,tikz-cd,tikz-3dplot,tkz-graph,pgfplots,ytableau,graphicx,float,rotating,hyperref,geometry,mathdots,savesym,cite, wasysym,amscd,graphicx,pifont,float,setspace,multirow,wrapfig,picture,subfigure, amsthm, enumitem}
\usepackage[normalem]{ulem} 
\usepackage[utf8]{inputenc}
\usepackage{prettyref}
\usepackage[T1]{fontenc}
\usepackage{datetime}
\usepackage{thm-restate}

\numberwithin{equation}{section}
\usetikzlibrary{arrows,positioning,decorations.pathmorphing,   decorations.markings, matrix, patterns}
\hypersetup{colorlinks=true}
\hypersetup{linkcolor=black}
\hypersetup{citecolor=black}
\hypersetup{urlcolor=black}
\theoremstyle{plain}
\newtheorem*{thm*}{Theorem}
\theoremstyle{plain}%
\newtheorem{thm}{Theorem}[section]
\newtheorem{conjec}[thm]{Conjecture}
\newtheorem{lem}[thm]{Lemma}

\theoremstyle{definition}
\newtheorem{defn}[thm]{Definition}

\newtheorem{rem}[thm]{Remark}

\newcommand{\tr}{{\bf tr} }

\numberwithin{equation}{section}

\tikzset{
  big arrow/.style={
    decoration={markings,mark=at position 1 with {\arrow[scale=1.5,#1]{>}}},
    postaction={decorate},
    shorten >=0.4pt},
  big arrow/.default=black}

\usetikzlibrary{fit}

\begin{document}

\begin{titlepage}
\begin{center}
\vspace{4cm}
{\Huge\bfseries  USp($4$)-models\\  }
\vspace{2cm}
{%
\LARGE  Mboyo Esole$^{\spadesuit}$ and Patrick Jefferson$^\diamondsuit$\\}
\vspace{1cm}

{\large $^{\spadesuit}$ Department of Mathematics, Northeastern University}\par
{\large  Boston, MA 02115, USA}\par

\vspace{.3cm}
{\large $^\diamondsuit$Center for Theoretical Physics, Department of Physics,\\
Massachusetts Institute of Technology
 }\par
{  Cambridge, MA 02139, U.S.A}\par
\vspace{3cm}
{ \bf{Abstract:}}\\
\end{center}
{\date{\today\  \currenttime}}

We study the geometry of elliptic fibrations satisfying the conditions of Step 2 of Tate's algorithm with a discriminant of valuation $4$. We call such geometries  USp($4$)-models, as the dual graph of their special fiber is the twisted affine Dynkin diagram $\widetilde{\text{C}}_2^t$. 
These geometries are used in string theory to model gauge theories with the non-simply-laced  Lie group USp($4$) on a smooth divisor $S$ of the base. 
Starting with a singular Weierstrass model of a USp($4$)-model,  we present a crepant resolution of its singularities. 
We study the fiber structure of this smooth elliptic fibration and identify the fibral divisors up to isomorphism as schemes over $S$.  
These are $\mathbb{P}^1$-bundles over $S$ or double covers of $\mathbb{P}^1$-bundles over $S$. 
We compute basic topological invariants such as the triple intersections of the fibral divisors and the Euler characteristic of the  USp($4$)-model. 
 In the case of Calabi-Yau threefolds, we also compute  the  Hodge numbers.  
 We study the compactfications of M/F theory on a USp($4$)-model  Calabi--Yau threefold.

\vfill 
{\today\par}


\end{titlepage}
\newpage 

\tableofcontents

\newpage 
\section{Introduction}

In the heart of F-theory is an algorithm that assigns  to  a given elliptic fibration a reductive Lie group $G$ and a representation $\mathbf{R}$ of $G$ \cite{Vafa:1996xn,Morrison:1996na,Morrison:1996pp,Bershadsky:1996nh, Heckman:2018jxk}.  
An elliptic fibration associated to a group $G$ via the F-theory algorithm is called a {\em $G$-model} \cite{Euler,G2}. 
$G$-models provide   a  beautiful bridge between the geometry of elliptic fibrations and aspects of gauge theories. 
The  Lie algebra $\mathfrak{g}$ of the group $G$ 
is determined by the types of singular fibers over  each irreducible component of  discriminant locus of the elliptic fibration.  
 The representation $\mathbf{R}$ is  encoded in the structure of singular fibers over  codimension-two points of the base corresponding to intersections of the components of the discriminant locus or  other singularities of the discriminant locus 
\cite{Aspinwall:1996nk,Morrison:2011mb,Marsano, GM1,G2}. 
 Furthermore, the global structure of the group $G$ depends on the Lie algebra $\mathfrak{g}$ and the Mordell--Weil group of the elliptic fibration \cite{Esole:2017hlw,Esole:2014dea}.  

The study of $G$-models in F-theory is a beautiful playground  for both mathematicians and physicists that has evolved tremendously in the last few years \cite{Anderson:2017zfm, EJJN1, EJJN2, Euler,Char1, Char2,Box,Morrison:2012np}. 
For example, a better understanding of $G$-models and  the flops connecting them has emerged from the construction of explicit crepant resolutions of Weierstrass models and  studies of the relative extended  K\"ahler cone over Weierstrass models \cite{G2,F4,SU2SU3,SU2G2,E7,ES,EY}.  Additionally, the  introduction of modern  intersection theory techniques as a method to compute topological invariants and characteristic numbers of $G$-models has turned conjectures into theorems and replaced painful  archaic computational methods with elegant algebraic techniques that not only apply to Calabi--Yau threefolds and fourfolds, but also provide results that apply to the broader setting of elliptic $n$-folds that are not necessarily Calabi--Yau  
\cite{AE1,AE2,Euler,Char1,Char2}. 

Despite this progress, there are still many open questions to address regarding the geometry of $G$-models. $G$-models of small rank have been studied in detail with the exception of one that has fallen between the cracks of recent investigations, namely the USp($4$)-model, which will be the focus of this paper.
Specifically, the $G$-models corresponding to Spin($n$) groups up to $n=8$ and all the 
 semi-simple compact Lie groups of rank up to three (with the exception of USp($4$)) have been studied in great detail (we recall that USp($4$) is isomorphic to Spin($5$) and has rank two\footnote{Spin groups of rank up to three are characterized by the following accidental isomorphisms: 
$$
 \text{Spin($3$) $\cong$ SU($2$), \  Spin($4$) $\cong$ SU($2$)$\times$SU($2$), \    Spin($5)$ $\cong$ USp($4$), 
 \   
 Spin($6$) $\cong$ SU($4$).
  }
$$
There are  only four semi-simple compact Lie groups of rank two, namely 
$$
\text{Spin}(4)\cong \text{SU}(2)\times\text{SU}(2), \quad \text{SU}(3),  \quad \text{Spin}(5)\cong\text{USp}(4), \quad \text{G}_2.
 $$
 The simply connected collisions of rank up to three are:
$$
  \text{SU($2$)$\times$SU($2$)},\quad 
 \text{SU($2$)$\times$SU($3$)},\quad  \text{SU($2$)$\times$USp($4$)},\quad   \text{SU($2$)$\times$G$_2$}.
$$
  The group USp($4$) contains Spin($4$) as a maximal subgroup and thus also  SU($2$)$\times$U($1$): 
$$
 \text{SU($2$)$\times$U($1$)}\subset \text{SU($2$)$\times$SU($2$)}\subset\text{USp($4$)}.
$$
 }.)   Spin($8$), G$_2$,  and Spin($7$)-models are studied in  \cite{G2}, Spin($6$) $\cong$ SU($4$), Spin($3$) $\cong$ SU($2$), and SU($3$) are studied in \cite{ESY1,ES}, and a SU($2$)$\times$U($1$)-model is studied in \cite{SU2U1}. The  Spin($4$) $\cong$ SU($2$)$\times$SU($2$) model is the only Spin-model for which $G$ is not simple but properly semi-simple; this model corresponds to a collision of two Kodaira fibers with dual graphs of type $\widetilde{\text{A}}_1$ and is studied in \cite{SO4}. 
The SO($3$), SO($5$), and SO($6$)-models are studied in \cite{SO}. The E$_7$-model is studied in \cite{E7}, the SU($5$)-model in \cite{EY,ESY2}, the SU($2$)$\times$ SU($3$)-model in \cite{SU2SU3}, the SU($2$)$\times$ G$_2$-model in \cite{SU2G2}.

The USp($4$)-model has a long history dating back to the early days of F-theory as it was instrumental in identifying that the $G$-models corresponding to elliptic fibrations with a fiber of type   I$^{\text{ns}}_{2n}$ or I$^{\text{ns}}_{2n+1}$ over the generic point of  their discriminant   locus were USp($2n$)-models \cite{Aspinwall:2000kf}. 
Note that a decorated Kodaira fiber of type I$_{2n}^{\text{ns}}$  or I$_{2n+1}^{\text{ns}}$ in both cases has a dual graph of type  $\widetilde{\text{C}}_{n}^t$; however, the geometric fibers of these two cases differ---respectively, the geometric fibers are of type  $\widetilde{\text{A}}_{2n-1}$  and  $\widetilde{\text{A}}_{2n}$. 
Since fibers of type I$_{2n+1}^{\text{ns}}$ have $\mathbb{Q}$-factorial terminal singularities, in this paper we define a  USp($4$)-model by means of a type I$_4^{\text{ns}}$ fiber.

 Carefully studying the geometry of the USp($4$)-model is a rewarding exercice as  the USp($4$)-model illustrates many of the subtleties of $G$-models while keeping the algebraic technicality at a minimum. In particular, the USp($4$)-model exhibits the following features:

\begin{enumerate}
\item {\bf Absence of flops}. The generic USp($4$)-model does not have any flops over the Weierstrass model \cite{EJJN1} as it is also the case for the  G$_2$-model \cite{G2},  the F$_4$-model \cite{F4}, and the generic USp($2n$)-model \cite{EJJN1}. 
 A smooth USp($4$)-model can be obtained by a crepant resolution of
 a singular Weierstrass model  by using just two blowups with smooth centers. 
 It follows that the geometry of the USp($4$)-model is particularly easy to study.
 
 \item  {\bf Non-split fibers without terminal $\mathbb{Q}$-factorial singularities}. $G$-models with fibers of type I$_{2n+1}^{\text{ns}}$ are known to have $\mathbb{Q}$-factorial terminal singularities. The USp($4$)-model with a fiber of type I$_4^{\text{ns}}$ is free of such problems and can therefore be defined as a flat fibration obtained by a  crepant resolution of a singular Weierstrass model.
\item {\bf Pseudo-real representations and half-hypermultiplets}. The fundamental representation of USp($2n$) is pseudo-real and can be carried by half-hypermultiplets.  
Other models with pseudo-real representations are the E$_7$-models with matter in the representation $\bf{56}$ \cite{E7,Box} and 
Spin($n$)-models for $n=3,4,5$ mod $8$ with matter in  irreducible  symplectic Majorana spinor representations.\footnote{
Symplectic Majorana spinors in Euclidean signature exist in dimension $3$, $4$, $5$ mod $8$; in dimension $4$ mod $8$, a symplectic Majorana spinor is not irreducible as it can be decomposed into two irreducible symplectic Majorana--Weyl spinors \cite[Table 2]{VanProeyen:1999ni}.}

 We recall that SU($2$) is isomorphic to USp($2$) so that the fundamental representation of SU($2$) is a  pseudo-real representation as a particular case of the fundamental representation of  USp($2n$). 
\item {\bf The simplest counter-example to the Katz--Vafa method}. The USp($4$)-model is  one of the simplest examples of a $G$-model for which the  Katz--Vafa method \cite{KV}  fails to correctly predict the matter representation  associated to a $G$-model in F-theory. 
Other examples are the Spin($7$)-model and the  SU($2$)$\times$G$_2$-model \cite{SU2G2}. For more example, see \cite{EK.KV}.
 
 \item  {\bf Counter-example to intersecting brane counting of matter multiplets in F-theory}. The USp($4$)-model is an example of a $G$-model for which the counting of charged hypermultiplets  does not follow the typical intersecting brane picture.  
 In intersecting brane models, the number of charged matter at the intersection of two branes is given by the intersection number of the two branes. 
\item {\bf Frozen representations.}
As a direct consequence of the previous point, we can have frozen representations \cite{F4,G2}: absence of hypermultiplets charged under a given representation while there are rational curves along the fiber that carry the weights of a representation for which there are no hypermultiplets charged. 
\end{enumerate}

They are other interesting questions that are not addressed here, such as the weak coupling limit of the theory and its tadpole cancellations \cite{Esole:2012tf,EFY,CDE}.

\subsection{Definition}

A USp($2n$)-model can be defined using a fiber of either type  I$_{2n}^{\text{ns}}$ or  I$_{2n+1}^{\text{ns}}$ as they both have the  affine Dynkin diagram of type $\widetilde{\text{C}}_n^{t}$ as a dual graph. 
However, since we care about the existence of a crepant resolution it is safer to consider only the I$_{2n}^{\text{ns}}$ case as 
models with I$_{2n+1}^{\text{ns}}$ fibers  generically suffer from $\mathbb{Q}$-factorial terminal singularities, which are lethal for the existence of a crepant resolution. In this spirit, we use the following restricted definition. 

  \begin{defn}
 A USp($2n$)-model is an elliptic fibration over a base $B$ 
 with a choice of a divisor $S$ that is an irreducible component of the discriminant locus such that 
 \begin{enumerate}
\item  The  fiber over the generic point of $S$ is of type  I$_{2n}^{\text{ns}}$. 
\item  The fiber over the generic point  of any other components of the discriminant locus is of  Kodaira type I$_1$ or II. 
 \end{enumerate}
  \end{defn}

A complete analysis of the Weierstrass model of a USp($2n$)-model should take into account the different options for the valuations of the coefficients of the Weierstrass equations:
\begin{equation}
a_1=a_3=0, \quad \nu(a_2)=0, \quad \nu(a_4)\geq n, \quad \nu(a_6)\geq 2n , \quad \nu(\Delta)=2n. 
\end{equation}

 One could think of such data as describing different types of complex structures for the elliptic fibration and their relevance can be appreciated by observing that they  affect the fiber structure of the elliptic fibration in the manner summarized in Figures \ref{Fig.V1}, \ref{Fig.V2}, and \ref{Fig.V3}.
The fiber of type I$_n^{\text{ns}}$ is defined by Step 2 of Tate's algorithm, which characterizes fibers with multiplicative reduction. 
We use the original Tate's algorithm and not one of the tables usually presented in the  F-theory literature as they miss several important subtleties. We define the USp($4$)-model by means of the following Weierstrass model\footnote{Let $\nu$ be the valuation defined by the smooth divisor $S$ supporting the fiber I$_n^{{\text{ns}}}$.  After an appropriate translation, the Weierstrass coefficients  are such that $\nu(a_3)$, $\nu(a_4)$, and $\nu(a_6)\geq 0$, $\nu(b_2)=0$,  and $\nu(\Delta)=n$. 
The non-split condition is the statement that $b_2$ is not a perfect square modulo $s$ (meaning in the residue field of the generic point of $S$, $b_2$ is not a perfect square).  We complete the square in $y$ to effectively have $a_1=a_3=0$.}:
\begin{equation}
\mathscr{E}_0: zy^2-(x^3 + a_2 x^2z + a_{4,2+\alpha} s^{2+\alpha} x z^2+ a_{6,4+\beta} s^{4+\beta} z^3)=0, \quad \alpha\beta=0, \quad \alpha,\beta \in \mathbb{Z}_{\geq 0},
\end{equation}
where we assume that the coefficients $a_2$, $a_{4,2+\alpha}$, $a_{6,4+\beta}$ are generic. 
The condition $\alpha\beta=0$ is necessary for the valuation of the discriminant  to be  $4$ as required for a fiber of type I$_4$. 
Since the coefficients are generic, a$_2$ is not a perfect square modulo $s$ and we have a fiber of type I$_n^{{\text{ns}}}$.
We distinguish three cases depending on the values of $\alpha$, and $\beta$:
 \begin{equation}
 (\alpha=0,\beta=0), \quad (\alpha=0,\beta>0), \quad (\alpha>0,\beta=0).
 \end{equation} 
One can think  of the first one as the most general one, and the two others as specializations obtained by increasing the valuation of $a_4$ or $a_6$ without changing the valuation of the discriminant. The valuation of $a_4$ and $a_6$ cannot increase simultaneously as otherwise it will force an increase in the valuation of the discriminant locus and the fiber will no longer be of type I$_4$. The previous statement explains why we have three cases to consider. 
 These three cases have very different fiber  structures as summarized in Figures \ref{Fig.V1}, \ref{Fig.V2}, and \ref{Fig.V3}. The first case ($\alpha=\beta=0$) is the one usually presented in F-theory tables.

\subsection{Matter representations}
 The Lie group USp($4$) is the unique compact, connected, and simply connected complex Lie group with Lie algebra of type C$_2$. 
 It has rank two and real dimension ten.  
Its center is the cyclic group $\mathbb{Z}/2\mathbb{Z}$ generated by minus the identity element of USp($4$). The quotient of USp($4$) by its center  is isomorphic to the orthogonal group SO($5$) since  USp($4$)  is isomorphic to Spin($5$). 
 Both  USp($4$) and SO($5$) have the same Lie algebra since the Lie algebra of type C$_2$ is isomorphic to the Lie algebra of type B$_2$. 
 The two fundamental representations of C$_2$ have respectively dimension $4$ and $5$ and correspond respectively to the spinor and vector representations of B$_2$.

When we look at  USp($4$) as the group Spin($5$), the relevant representations are the adjoint representation of dimension $10$, the spinor representation of dimension $4$, and the vector representation of dimension $5$. 
The spinor and vector representations of Spin($5$) correspond respectively to the fundamental of USp($4$) and the traceless antisymmetric two-tensor representation of USp($4$). 
We use the physics convention of labeling representations by their dimensions in bold characters, as this will be sufficient to distinguish the representations we deal with in this paper. 
Thus, we denote the fundamental, the traceless antisymmetric 
 (the trace is defined with respect to the antisymmetric symplectic metric), and the adjoint representations of C$_2$  by (respectively) $\mathbf{4}$, $\mathbf{5}$, and $\mathbf{10}$. The representations $\bf{4}$, $\bf{5}$, and $\bf{10}$  of C$_2$ correspond to (respectively) the spinor, vector, and adjoint representation of $\mathfrak{so}_5$.  
The representation $\bf{4}$ is minuscule and pseudo-real while the representation $\bf{5}$ is quasi-minuscule and real.
The representation $\bf{10}$, being the adjoint representation, is always real.

The matter representation \textbf{R} associated to a USp($4$)-model is the sum of the above three irreducible representations: 
\begin{equation}
\bf{R}=\mathbf{4}\oplus\mathbf{5}\oplus\mathbf{10}.
\end{equation}

The fiber  I$_4^{\text{ns}}$  degenerates in codimension two producing vertical curves whose weights correspond to the representations $\mathbf{4}$, $\mathbf{5}$, and $\mathbf{10}$ of C$_2$. The representations $\bf{4}$ and $\bf{5}$ are the two fundamental representations of the Lie algebra  C$_2$ in the sense that their highest weights are the simple roots of C$_2$. 
Moreover, all the weights of the representations $\bf{4}$ and $\bf{5}$ are also weights of the adjoint representation $\bf{10}$. 
This fact explains why the USp($4$)-model has a unique crepant resolution.

The first codimension-two degeneration is 
    when one of the nodes of I$_4$ splits into two rational curves, producing weights in the representation $\bf{4}$. 
    Although the corresponding fiber is geometrically of Kodaira type I$_5$, the arithmetic 
  type of the fiber depends on the dimension of the base:
\begin{equation}
\begin{cases}
 \text{I}^{\text{ns}}_4 \to \text{I}^{\text{ns}}_5 \quad\text{if}\quad \dim\ B\geq 3\longrightarrow \text{no matter}, \\
  \text{I}^{\text{ns}}_4 \to \text{I}^{\text{* ss}}_0 \quad \   \text{if}\quad \dim B \geq 2\  \longrightarrow  \   \text{Representation $\bf{5}$}, \\
 \text{I}^{\text{ns}}_4 \to \text{I}^{\text{s}}_5 \quad \   \text{if}\quad \dim B =2\  \longrightarrow \  \text{Representation $\bf{4}$}, \\
 \text{I}^{\text{ns}}_4 \to \text{I}^{\text{*}}_0 \quad \   \text{if}\quad \dim B =2\  \longrightarrow  \   \text{Representation $\bf{5}$}, 
 \end{cases}
 \end{equation}
The second  codimension-two degeneration is when the non-split  node of I$_4^{\text{ns}}$ collapses into  a double node producing weights in the representation $\bf{5}$.  The fiber is an incomplete Kodaira fiber of type I$_0^*$:
  \begin{equation}
  \text{I}_4^{\text{ns}} \to \text{I}^*_0.
  \end{equation}

\subsection{Counting charged hypermultiplets in 5D and  6D}
A compactification of M-theory on a smooth Calabi--Yau threefold yields a five-dimensional ${\cal N}=1$ supergravity theory.                                                                                                      
Such a theory has vector multiplets and hypermultiplets, and the kinetic terms of the vector fields as well as Chern--Simons terms are encoded in a cubic prepotential that corresponds to the triple intersection polynomial of the divisors of the Calabi--Yau manifold. 
The prepotential can have a quantum contribution coming from integrating out massive charged  hypermultiplets. 
It follows that the prepotential depends on the number of charged hypermultiplets and in the best case, a direct comparison with geometry provides a determination of the number of charged hypermultiplets using the intersection ring of the Calabi--Yau threefold.

In the case of a USp($4$)-model in which the fiber I$_4^{\text{ns}}$ is over the generic point of a divisor $S$ of self-intersection $S^2$ in the base $B$ and genus $g$, the number of hypermultiplets charged under the representations $\bf{10}$, $\bf{5}$, and $\bf{4}$ are respectively  
\begin{equation}
n_{\bf 10 }=g,\quad 	n_{\bf 5}= 1-g+S^2, \quad n_{\mathbf{4}} = 16(1 - g) +  4S^2.
	\end{equation} 
 In Section \ref{Sec:Blowup}, we present a resolution by  a succession of two blowups with smooth centers. 
In Section \ref{Sec:Euler}, we compute the  Euler characteristic and the triple intersection numbers of the fibral divisors using the method developed in \cite{Euler}.
In the Calabi-Yau case, we also compute the Hodge numbers of the USp($4$)-model as in \cite{Euler}. 
We will specialize to the Calabi-Yau case and the compactification of M-theory on such a Calabi-Yau threefold to a five-dimensional supergravity theory.  
We determine the representation associated with a USp($4$)-model by computing the intersection of the rational curves appearing when the fiber degenerates over codimension-two points. 
 While the representation $\mathbf{4}$ is always present for a base of dimension two or higher, the representation $\mathbf{5}$ appears only if the base is a surface. 
We also illustrate the phenomena of frozen representations in F-theory. A representation is said to be {\em frozen} when there are curves carrying the weights of the representation but no physical states are charged under that representation. 
 By comparing the triple intersection numbers of the fibral divisors to the one-loop prepotential of the five dimensional minimal supergravity theory with the same matter content, we compute 
the multiplicity of each representation. We also check that these multiplicities match those of an anomaly free six-dimensional theory with a compatible spectrum of charged hypermultiplets.

\section{Weierstrass model, crepant resolutions, and fiber degenerations}

\noindent{\bf Conventions.}  Given a simple Lie algebra with Dynkin diagram $\mathfrak{g}$, we denote by $\widetilde{\mathfrak{g}}$ the untwisted affine Dynkin diagram which reduces to $\mathfrak{g}$ after removing one of its node of multiplicity one. 
We denote by $\widetilde{\mathfrak{g}}^t$ the twisted affine Dynkin diagram  obtained from $\widetilde{\mathfrak{g}}$ by reversing the direction of all the arrows of the Dynkin diagram of $\widetilde{\mathfrak{g}}$. The Cartan matrix of $\widetilde{\mathfrak{g}}^t$ is the transpose of the Cartan matrix of  $\widetilde{\mathfrak{g}}$ \cite{G2,Euler,F4}; thus,
the difference between  $\widetilde{\mathfrak{g}}$ and  $\widetilde{\mathfrak{g}}^t$ only matters when $\mathfrak{g}$ is not simply laced. 
We denote by $V(f_1, \ldots, f_n)$ the vanishing locus $f_1=\ldots =f_n=0$. 

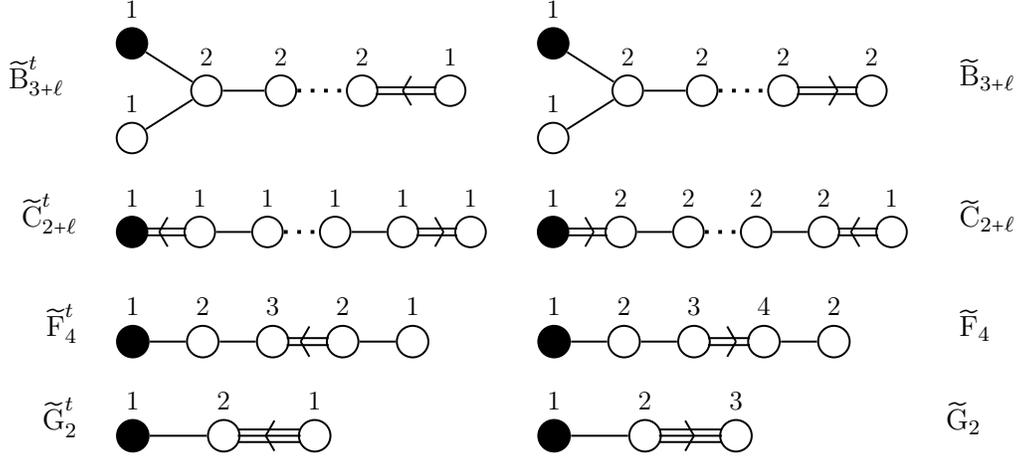
\begin{figure}[htb]
\begin{center}
 \scalebox{1}{$			
\begin{array}{r l  l l }

												   \begin{array}{c}
												  \\
						 \widetilde{\text{B}}_{3+\ell}^t  \\
						 \\
						 \end{array}
						 &\scalebox{.9}{$\begin{array}{c} \begin{tikzpicture}
				\node[draw,circle,thick,scale=1.25,fill=black,label=above:{1}] (1) at (-.1,.7){};
				\node[draw,circle,thick,scale=1.25,label=above:{1}] (2) at (-.1,-.7){};	
				\node[draw,circle,thick,scale=1.25,label=above:{2}] (3) at (1,0){};
				\node[draw,circle,thick,scale=1.25,label=above:{2}] (4) at (2.1,0){};
				\node[draw,circle,thick,scale=1.25,label=above:{2}] (5) at (3.3,0){};	
				\node[draw,circle,thick,scale=1.25,label=above:{1}] (6) at (4.6,0){};	
				\draw[thick] (1) to (3) to (4);
				\draw[thick] (2) to (3);
				\draw[ultra thick, loosely dotted] (4) to (5) {};
				\draw[thick] (3.5,-0.05) --++ (.9,0){};
				\draw[thick] (3.5,+0.05) --++ (.9,0){};
				\draw[thick]
					(3.9,0) --++ (60:.25)
					(3.9,0) --++ (-60:.25);
			\end{tikzpicture}\end{array}$}		&

\scalebox{.9}{$\begin{array}{c} \begin{tikzpicture}
				\node[draw,circle,thick,scale=1.25,fill=black,label=above:{1}] (1) at (-.1,.7){};
				\node[draw,circle,thick,scale=1.25,label=above:{1}] (2) at (-.1,-.7){};	
				\node[draw,circle,thick,scale=1.25,label=above:{2}] (3) at (1,0){};
				\node[draw,circle,thick,scale=1.25,label=above:{2}] (4) at (2.1,0){};
				\node[draw,circle,thick,scale=1.25,label=above:{2}] (5) at (3.3,0){};	
				\node[draw,circle,thick,scale=1.25,label=above:{2}] (6) at (4.6,0){};	
				\draw[thick] (1) to (3) to (4);
				\draw[thick] (2) to (3);
				\draw[ultra thick, loosely dotted] (4) to (5) {};
				\draw[thick] (3.5,-0.05) --++ (.9,0){};
				\draw[thick] (3.5,+0.05) --++ (.9,0){};
				\draw[thick]
					(4.1,0) --++ (120:.25)
					(4.1,0) --++ (-120:.25);
			\end{tikzpicture}\end{array}$}				
			&
			  \begin{array}{c}
												  \\
						 \widetilde{\text{B}}_{3+\ell}  \\
						 \\
						 \end{array}
			\\
			
 \widetilde{\text{C}}_{2+\ell}^t  
						   &
			\scalebox{.9}{$\begin{array}{c} \begin{tikzpicture}
				\node[draw,circle,thick,scale=1.25,fill=black,label=above:{1}] (1) at (-.2,0){};
				\node[draw,circle,thick,scale=1.25,label=above:{1}] (3) at (.8,0){};
				\node[draw,circle,thick,scale=1.25,label=above:{1}] (4) at (1.8,0){};
				\node[draw,circle,thick,scale=1.25,label=above:{1}] (5) at (2.8,0){};	
				\node[draw,circle,thick,scale=1.25,label=above:{1}] (6) at (3.8,0){};	
				\node[draw,circle,thick,scale=1.25,label=above:{1}] (7) at (4.8,0){};	
				\draw[thick]   (3) to (4);
				\draw[thick] (5) to (6);
				\draw[ultra thick, loosely dotted] (4) to (5) {};
				\draw[thick] (4.,-0.05) --++ (.6,0){};
				\draw[thick] (4.,+0.05) --++ (.6,0){};
								\draw[thick] (-.2,-0.05) --++ (.8,0){};
				\draw[thick] (-.2,+0.05) --++ (.8,0){};

				\draw[thick]
					(4.4,0) --++ (-120:.25)
					(4.4,0) --++ (120:.25);
					\draw[thick]
					(0.2,0) --++ (-60:.25)
					(0.2,0) --++ (60:.25);
			\end{tikzpicture}\end{array}$}
& 
 \scalebox{.9}{$\begin{array}{c} \begin{tikzpicture}
				\node[draw,circle,thick,scale=1.25,fill=black,label=above:{1}] (1) at (-.2,0){};
				\node[draw,circle,thick,scale=1.25,label=above:{2}] (3) at (.8,0){};
				\node[draw,circle,thick,scale=1.25,label=above:{2}] (4) at (1.8,0){};
				\node[draw,circle,thick,scale=1.25,label=above:{2}] (5) at (2.8,0){};	
				\node[draw,circle,thick,scale=1.25,label=above:{2}] (6) at (3.8,0){};	
				\node[draw,circle,thick,scale=1.25,label=above:{1}] (7) at (4.8,0){};	
				\draw[thick]   (3) to (4);
				\draw[thick] (5) to (6);
				\draw[ultra thick, loosely dotted] (4) to (5) {};
				\draw[thick] (4.,-0.05) --++ (.6,0){};
				\draw[thick] (4.,+0.05) --++ (.6,0){};
				\draw[thick] (-.2,-0.05) --++ (.8,0){};
				\draw[thick] (-.2,+0.05) --++ (.8,0){};

				\draw[thick]
					(4.2,0) --++ (-60:.25)
					(4.2,0) --++ (60:.25);
					\draw[thick]
					(0.4,0) --++ (-120:.25)
					(0.4,0) --++ (120:.25);
			\end{tikzpicture}\end{array}$}

			&
			  \begin{array}{c}
												  \\
						 \widetilde{\text{C}}_{2+\ell}  \\
						 \\
						 \end{array}

			\\
\widetilde{\text{F}}_4^t 
&			

\scalebox{.93}{$\begin{array}{c}\begin{tikzpicture}
				\node[draw,circle,thick,scale=1.25,fill=black,label=above:{1}] (1) at (0,0){};
				\node[draw,circle,thick,scale=1.25,label=above:{2}] (2) at (1,0){};
				\node[draw,circle,thick,scale=1.25,label=above:{3}] (3) at (2,0){};
				\node[draw,circle,thick,scale=1.25,label=above:{2}] (4) at (3,0){};
				\node[draw,circle,thick,scale=1.25,label=above:{1}] (5) at (4,0){};
				\draw[thick] (1) to (2) to (3);
				\draw[thick]  (4) to (5);
				\draw[thick] (2.2,0.05) --++ (.6,0);
				\draw[thick] (2.2,-0.05) --++ (.6,0);
				\draw[thick]
					(2.4,0) --++ (60:.25)
					(2.4,0) --++ (-60:.25);
			\end{tikzpicture}\end{array}$} & 
			\scalebox{.93}{$\begin{array}{c}\begin{tikzpicture}
				\node[draw,circle,thick,scale=1.25,fill=black,label=above:{1}] (1) at (0,0){};
				\node[draw,circle,thick,scale=1.25,label=above:{2}] (2) at (1,0){};
				\node[draw,circle,thick,scale=1.25,label=above:{3}] (3) at (2,0){};
				\node[draw,circle,thick,scale=1.25,label=above:{4}] (4) at (3,0){};
				\node[draw,circle,thick,scale=1.25,label=above:{2}] (5) at (4,0){};
				\draw[thick] (1) to (2) to (3);
				\draw[thick]  (4) to (5);
				\draw[thick] (2.2,0.05) --++ (.6,0);
				\draw[thick] (2.2,-0.05) --++ (.6,0);
				\draw[thick]
					(2.6,0) --++ (120:.25)
					(2.6,0) --++ (-120:.25);
			\end{tikzpicture}\end{array}$}
				&
				  \begin{array}{c}
												  \\
						 \widetilde{\text{F}}_{4}  \\
						 \\
						 \end{array}
				\\

			\widetilde{\text{G}}_2^t

 &			\scalebox{.93}{$\begin{array}{c}\begin{tikzpicture}
				\node[draw,circle,thick,scale=1.25,label=above:{1}, fill=black] (1) at (0,0){};
				\node[draw,circle,thick,scale=1.25,label=above:{2}] (2) at (1.3,0){};
				\node[draw,circle,thick,scale=1.25,label=above:{1}] (3) at (2.6,0){};
				\draw[thick] (1) to (2);
				\draw[thick] (1.5,0.09) --++ (.9,0);
				\draw[thick] (1.5,-0.09) --++ (.9,0);
				\draw[thick] (1.5,0) --++ (.9,0);
				\draw[thick]
					(1.9,0) --++ (60:.25)
					(1.9,0) --++ (-60:.25);
			\end{tikzpicture}\end{array}
		$} 		& 
		\scalebox{.93}{$\begin{array}{c}\begin{tikzpicture}
				\node[draw,circle,thick,scale=1.25,label=above:{1}, fill=black] (1) at (0,0){};
				\node[draw,circle,thick,scale=1.25,label=above:{2}] (2) at (1.3,0){};
				\node[draw,circle,thick,scale=1.25,label=above:{3}] (3) at (2.6,0){};
				\draw[thick] (1) to (2);
				\draw[thick] (1.5,0.09) --++ (.9,0);
				\draw[thick] (1.5,-0.09) --++ (.9,0);
				\draw[thick] (1.5,0) --++ (.9,0);
				\draw[thick]
					(2,0) --++ (120:.25)
					(2,0) --++ (-120:.25);
			\end{tikzpicture}\end{array}
		$} 
&

						 \widetilde{\text{G}}_{2}  				\\\end{array}$}
										\end{center}	
												 \caption{ Twisted affine Lie algebras vs affine Lie algebras for non-simply laced cases. In each case, the colored  node is the affine node and  the number of non-colored node is given by the index. 												 Only those on the left appear in the theory of elliptic fibrations as dual graphs of  the fiber over the generic point of an irreducible component of the discriminant locus. 
		In the notation of Kac \cite{Kac},   $\widetilde{\text{B}}^t_{3+\ell}$,    $\widetilde{\text{C}}^t_{1+\ell}$,     $\widetilde{\text{F}}^t_{4}$  , and $\widetilde{\text{G}}^t_{2}$  correspond respectively to 
		${\text{A}}^{(2)}_{2\ell-1}$,    ${\text{D}}^{(2)}_{1+\ell}$,     ${\text{E}}^{(2)}_{6}$,  and ${\text{D}}^{(3)}_{4}$.
												  \label{Fig:AffineLieAlgebras}}
						  \end{figure}

\subsection{Weierstrass model}
The following elliptic fibration is the generic Weierstrass model for a USp(4) model with  a fiber I$_4^{\text{ns}}$ over the Cartier divisor $S=V(s)$:
\begin{equation}
\mathscr{E}_0: zy^2=x^3 + a_2 x^2z + a_{4,2} s^2 x z^2+ a_{6,4} s^4z^3.
\end{equation}
The short  Weierstrass equation and the discriminant  are characterized by the invariants:
\begin{align}
& c_4=16(a_2^2-3 a_{4,2} s^2), \\
& c_6=32(2 a_2^3-9 a_2 a_{4,2} s^2+27 a_{6,4} s^{4}),\\
& \Delta=-16 s^4 \Delta', \quad \Delta'=4 a_2^3 a_{6,4}-a_2^2 a_4^2-18 a_2 a_{4,2} a_{6,4} s^2+4 a_{4,2}^3 s^2+27 a_{6,4}^2 s^{4}.
\end{align}
The fiber over $S=V(s)$ is of type I$_4^{\text{ns}}$ with valuation $(0,0,4)$ over $(c_4,c_6,\Delta)$ and $b_2$ not a perfect square modulo $s$. Generically, the Mordell-Weil group is trivial. 
To ensure that the fiber is of type I$_4^{\text{ns}}$, we should have $v(a_2)=0$ and $v(a_{4})=2$ or $v(a_{6})=4$. 
There are  three cases to consider depending on the  valuation of $a_4$ and $a_{6,4}$: 
\begin{align}
v(a_4)=2, v(a_6)=4,\\
v(a_4)>2, v(a_6)=4,\\
v(a_4)=2, v(a_6)>4.
\end{align}

The reduced discriminant is composed of two irreducible components intersecting non-transversally along the reducible locus
\begin{equation}
s=4 a_2^2(4a_2 a_6- a_{4}^2)=0.
\end{equation}
This is the codimension-two locus over which the generic fiber over $S$ degenerates;
it consists of two irreducible components. In the generic case, these two loci are  
\begin{equation}
S\cap V(a_2), \quad S\cap V(a_4^2-a_2 a_6).
\end{equation}
The locus $S\cap V(a_2)$ is the intersection of $S$ with the  cuspidal locus $V(c_4,c_6)$ of the elliptic fibration. 
Taking into account the three possible cases of valuations for $a_4$ and $a_6$, we have: 
\begin{align}
S\cap \Delta'~:~
\begin{cases}
v(a_4)=2, v(a_6)=4 & \Longrightarrow   s= a_2^2(4a_2 a_{6,4}- a_{4,2}^2)=0,\\
v(a_4)>2, v(a_6)=4 & \Longrightarrow   s=a_2^3 a_{6,4}=0,\\
v(a_4)=2, v(a_6)>4 & \Longrightarrow   s= a_2^2 a_{4,2}^2=0.
\end{cases}
\end{align}

\begin{table}[htb]
\begin{center}
$			
\begin{array}{|c|c| c  |} \hline
\vrule width 0pt height 3ex 
 \text{Fiber Type} & \text{ Dual graph  } & \text{Dual graph of Geometric fiber } \\\hline
 			  \begin{array}{c}
\text{I}^{\text{ns}}_{3}, \text{IV}^{\text{ns}}\\
						  \\
						 \widetilde{\text{A}}_{1}  \\
						 \\
						 \end{array}

 &			\scalebox{1}{$\begin{array}{c}\begin{tikzpicture}
				\node[draw,circle,thick,scale=1.25,label=above:{1}, fill=black] (1) at (0,0){};
				\node[draw,circle,thick,scale=1.25,label=above:{1}] (2) at (1.3,0){};
				\draw[thick] (0.15,0.1) --++ (.95,0);
				\draw[thick] (0.15,-0.09) --++ (.95,0);
				
			\end{tikzpicture}\end{array}
		$} &
		\scalebox{1}{$\begin{array}{c}\begin{tikzpicture}
				\node[draw,circle,thick,scale=1.25,label=above:{1}, fill=black] (1) at (-.6,0){};
				\node[draw,circle,thick,scale=1.25,label=above:{1}] (2) at (45:.6){};
				\node[draw,circle,thick,scale=1.25,label=below:{1}] (3) at (-45:.6){};
				\draw[thick] (0,0) --(1);
				\draw[thick] (0,0) --(2);
				\draw[thick] (0,0) --(3);
								\draw[<->,>=stealth',semithick,dashed]  (.8,-.5) arc (-30:30:1cm);
			\end{tikzpicture}\end{array}$}

				\\\hline

						 \begin{array}{c}

						  \text{I}^{*\text{ns}}_{\ell-3}\\
						  \\
						 \widetilde{\text{B}}_{\ell}^t  \\
						 \\
						 (\ell\geq 3)
						 \end{array}
						 &\scalebox{1}{$\begin{array}{c} \begin{tikzpicture}
				\node[draw,circle,thick,scale=1.25,fill=black,label=above:{1}] (1) at (-.1,.7){};
				\node[draw,circle,thick,scale=1.25,label=above:{1}] (2) at (-.1,-.7){};	
				\node[draw,circle,thick,scale=1.25,label=above:{2}] (3) at (1,0){};
				\node[draw,circle,thick,scale=1.25,label=above:{2}] (4) at (2.1,0){};
				\node[draw,circle,thick,scale=1.25,label=above:{2}] (5) at (3.3,0){};	
				\node[draw,circle,thick,scale=1.25,label=above:{1}] (6) at (4.6,0){};	
				\draw[thick] (1) to (3) to (4);
				\draw[thick] (2) to (3);
				\draw[ultra thick, loosely dotted] (4) to (5) {};
				\draw[thick] (3.5,-0.05) --++ (.9,0){};
				\draw[thick] (3.5,+0.05) --++ (.9,0){};
				\draw[thick]
					(3.9,0) --++ (60:.25)
					(3.9,0) --++ (-60:.25);
			\end{tikzpicture}\end{array}$}&
			
\scalebox{1}{$\begin{array}{c} \begin{tikzpicture}
				\node[draw,circle,thick,scale=1.25,fill=black,label=above:{1}] (1) at (-.1,.7){};
				\node[draw,circle,thick,scale=1.25,label=above:{1}] (2) at (-.1,-.7){};	
				\node[draw,circle,thick,scale=1.25,label=above:{2}] (3) at (1,0){};
				\node[draw,circle,thick,scale=1.25,label=above:{2}] (4) at (2.1,0){};
				\node[draw,circle,thick,scale=1.25,label=above:{2}] (5) at (3.3,0){};	
				\node[draw,circle,thick,scale=1.25,label=above:{1}] (6) at (4.6,.7){};	
								\node[draw,circle,thick,scale=1.25,label=above:{1}] (7) at (4.6,-.7){};	
				\draw[thick] (1) to (3) to (4);
				\draw[thick] (2) to (3);
				\draw[ultra thick, loosely dotted] (4) to (5) {};
				\draw[thick] (5) to (6); 
				\draw[thick] (5) to (7);
				\draw[<->,>=stealth',semithick,dashed]  (5,-.5) arc (-30:30:1.2cm);
			\end{tikzpicture}\end{array}$}
			
			\\\hline
			\begin{array}{c}
			     \text{I}_{2\ell+2}^{\text{ns}}
			\\
			\\
			\vrule width 0pt height 3ex 
 \widetilde{\text{C}}_{\ell+1}^t  \\
						 \\
						 (\ell\geq 1)
 \end{array}
   &
			\scalebox{.95}{$\begin{array}{c} \begin{tikzpicture}
				\node[draw,circle,thick,scale=1.25,fill=black,label=above:{1}] (1) at (-.2,0){};
				\node[draw,circle,thick,scale=1.25,label=above:{1}] (3) at (.8,0){};
				\node[draw,circle,thick,scale=1.25,label=above:{1}] (4) at (1.8,0){};
				\node[draw,circle,thick,scale=1.25,label=above:{1}] (5) at (2.8,0){};	
				\node[draw,circle,thick,scale=1.25,label=above:{1}] (6) at (3.8,0){};	
				\node[draw,circle,thick,scale=1.25,label=above:{1}] (7) at (4.8,0){};	
				\draw[thick]   (3) to (4);
				\draw[thick] (5) to (6);
				\draw[ultra thick, loosely dotted] (4) to (5) {};
				\draw[thick] (4.,-0.05) --++ (.6,0){};
				\draw[thick] (4.,+0.05) --++ (.6,0){};
								\draw[thick] (-.2,-0.05) --++ (.8,0){};
				\draw[thick] (-.2,+0.05) --++ (.8,0){};

				\draw[thick]
					(4.4,0) --++ (-120:.25)
					(4.4,0) --++ (120:.25);
					\draw[thick]
					(0.2,0) --++ (-60:.25)
					(0.2,0) --++ (60:.25);
			\end{tikzpicture}\end{array}$}

			&

			\scalebox{.95}{$\begin{array}{c} \begin{tikzpicture}
				\node[draw,circle,thick,scale=1.25,fill=black,label=above:{1}] (1) at (-.2,0){};	
				\node[draw,circle,thick,scale=1.25,label=below:{1}] (3a) at (.8,-.8){};
				\node[draw,circle,thick,scale=1.25,label=below:{1}] (4a) at (1.8,-.8){};
				\node[draw,circle,thick,scale=1.25,label=below:{1}] (5a) at (2.8,-.8){};	
				\node[draw,circle,thick,scale=1.25,label=below:{1}] (6a) at (3.8,-.8){};	
								\node[draw,circle,thick,scale=1.25,label=above:{1}] (3b) at (.8,.8){};
				\node[draw,circle,thick,scale=1.25,label=above:{1}] (4b) at (1.8,.8){};
				\node[draw,circle,thick,scale=1.25,label=above:{1}] (5b) at (2.8,.8){};	
				\node[draw,circle,thick,scale=1.25,label=above:{1}] (6b) at (3.8,.8){};	

				\node[draw,circle,thick,scale=1.25,label=above:{1}] (7) at (4.8,0){};	
				\draw[thick]   (4b)--(3b)--(1)--(3a)--(4a);
				\draw[thick]   (5b)--(6b)--(7)--(6a)--(5a);
				\draw[ultra thick, loosely dotted] (4a) to (5a) {};
								\draw[ultra thick, loosely dotted] (4b) to (5b) {};
\draw[<->,>=stealth',semithick,dashed] ($(4a)+(0,0.3)$) --($(4b)-(0,0.3)$) {};
\draw[<->,>=stealth',semithick,dashed] ($(5a)+(0,0.3)$) --($(5b)-(0,0.3)$) {};
\draw[<->,>=stealth',semithick,dashed] ($(6a)+(0,0.3)$) --($(6b)-(0,0.3)$) {};
\draw[<->,>=stealth',semithick,dashed] ($(3a)+(0,0.3)$) --($(3b)-(0,0.3)$) {};
			\end{tikzpicture}\end{array}$}

			\\\hline

			\begin{array}{c}
			     \text{I}_{2\ell+3}^{\text{ns}}
			\\
			\\
			\vrule width 0pt height 3ex 
 \widetilde{\text{C}}_{\ell+1}^t   \\
						 \\
						 (\ell\geq 1)
 \end{array}
   &
			\scalebox{.95}{$\begin{array}{c} \begin{tikzpicture}
				\node[draw,circle,thick,scale=1.25,fill=black,label=above:{1}] (1) at (-.2,0){};
				\node[draw,circle,thick,scale=1.25,label=above:{1}] (3) at (.8,0){};
				\node[draw,circle,thick,scale=1.25,label=above:{1}] (4) at (1.8,0){};
				\node[draw,circle,thick,scale=1.25,label=above:{1}] (5) at (2.8,0){};	
				\node[draw,circle,thick,scale=1.25,label=above:{1}] (6) at (3.8,0){};	
				\node[draw,circle,thick,scale=1.25,label=above:{1}] (7) at (4.8,0){};	
				\draw[thick]   (3) to (4);
				\draw[thick] (5) to (6);
				\draw[ultra thick, loosely dotted] (4) to (5) {};
				\draw[thick] (4.,-0.05) --++ (.6,0){};
				\draw[thick] (4.,+0.05) --++ (.6,0){};
								\draw[thick] (-.2,-0.05) --++ (.8,0){};
				\draw[thick] (-.2,+0.05) --++ (.8,0){};

				\draw[thick]
					(4.4,0) --++ (-120:.25)
					(4.4,0) --++ (120:.25);
					\draw[thick]
					(0.2,0) --++ (-60:.25)
					(0.2,0) --++ (60:.25);
			\end{tikzpicture}\end{array}$}

			&
			
				\scalebox{.95}{$\begin{array}{c} \begin{tikzpicture}
				\node[draw,circle,thick,scale=1.25,fill=black,label=above:{1}] (1) at (-.2,0){};	
				\node[draw,circle,thick,scale=1.25,label=below:{1}] (3a) at (.8,-.8){};
				\node[draw,circle,thick,scale=1.25,label=below:{1}] (4a) at (1.8,-.8){};
				\node[draw,circle,thick,scale=1.25,label=below:{1}] (5a) at (2.8,-.8){};	
				\node[draw,circle,thick,scale=1.25,label=below:{1}] (6a) at (3.8,-.8){};	
								\node[draw,circle,thick,scale=1.25,label=above:{1}] (3b) at (.8,.8){};
				\node[draw,circle,thick,scale=1.25,label=above:{1}] (4b) at (1.8,.8){};
				\node[draw,circle,thick,scale=1.25,label=above:{1}] (5b) at (2.8,.8){};	
				\node[draw,circle,thick,scale=1.25,label=above:{1}] (6b) at (3.8,.8){};	

				\node[draw,circle,thick,scale=1.25,label=below:{1}] (7a) at (4.8,-.8){};
				\node[draw,circle,thick,scale=1.25,label=above:{1}] (7b) at (4.8,.8){};	
				\draw[thick]   (4b)--(3b)--(1)--(3a)--(4a);
				\draw[thick]   (5b)--(6b)--(7b)--(7a)--(6a)--(5a);
				\draw[ultra thick, loosely dotted] (4a) to (5a) {};
								\draw[ultra thick, loosely dotted] (4b) to (5b) {};

\draw[<->,>=stealth',semithick,dashed] ($(4a)+(0,0.3)$) --($(4b)-(0,0.3)$) {};
\draw[<->,>=stealth',semithick,dashed] ($(5a)+(0,0.3)$) --($(5b)-(0,0.3)$) {};
\draw[<->,>=stealth',semithick,dashed] ($(6a)+(0,0.3)$) --($(6b)-(0,0.3)$) {};
\draw[<->,>=stealth',semithick,dashed] ($(3a)+(0,0.3)$) --($(3b)-(0,0.3)$) {};
\draw[<->,>=stealth',semithick,dashed] (5.2,-1) arc (-70:70:1) {};
			\end{tikzpicture}\end{array}$}

			\\\hline
			\begin{array}{c}
			\text{IV}^{* \text{ns}}
			\\ 
			\vrule width 0pt height 3ex 
\widetilde{\text{F}}_4^t      
\end{array}
&			

\scalebox{1}{$\begin{array}{c}\begin{tikzpicture}
				\node[draw,circle,thick,scale=1.25,fill=black,label=above:{1}] (1) at (0,0){};
				\node[draw,circle,thick,scale=1.25,label=above:{2}] (2) at (1,0){};
				\node[draw,circle,thick,scale=1.25,label=above:{3}] (3) at (2,0){};
				\node[draw,circle,thick,scale=1.25,label=above:{2}] (4) at (3,0){};
				\node[draw,circle,thick,scale=1.25,label=above:{1}] (5) at (4,0){};
				\draw[thick] (1) to (2) to (3);
				\draw[thick]  (4) to (5);
				\draw[thick] (2.2,0.05) --++ (.6,0);
				\draw[thick] (2.2,-0.05) --++ (.6,0);
				\draw[thick]
					(2.4,0) --++ (60:.25)
					(2.4,0) --++ (-60:.25);
			\end{tikzpicture}\end{array}$} &
			\scalebox{1}{$\begin{array}{c}\begin{tikzpicture}
				\node[draw,circle,thick,scale=1.25,fill=black,label=below:{1}] (0) at (0,0){};
				\node[draw,circle,thick,scale=1.25,label=below:{2}] (1) at (.8,0){};
				\node[draw,circle,thick,scale=1.25,label=below:{3}] (2) at (.8*2,0){};
				\node[draw,circle,thick,scale=1.25,label=below:{2}] (3) at (.8*3,0){};
				\node[draw,circle,thick,scale=1.25,label=below:{1}] (4) at (.8*4,0){};
								\node[draw,circle,thick,scale=1.25,label=left:{2}] (5) at (.8*2,.8*1){};
																\node[draw,circle,thick,scale=1.25,label=above:{1}] (6) at (.8*2,.8*2){};
				\draw[thick] (0)--(1)--(2)--(3)--(4);
				\draw[thick] (2)--(5)--(6);
				\draw[<->,>=stealth',semithick,dashed]  (2.5,0.3) arc (25:65:1.3cm);
				\draw[<->,>=stealth',semithick,dashed]  (3.3,0.3) arc (25:65:3cm);

			\end{tikzpicture}\end{array}$}
		
			\\\hline
			\begin{array}{c}
			\text{I}^{*\text{ss}}_{0}
			\\
			\widetilde{\text{B}}_3^t
\end{array}
 &			\scalebox{1}{$\begin{array}{c}\begin{tikzpicture}
				\node[draw,circle,thick,scale=1.25,label=below:{1}] (1) at (0,-.5){};
				\node[draw,circle,thick,scale=1.25,label=above:{2}] (2) at (1.3,0){};
				\node[draw,circle,thick,scale=1.25,label=above:{1}] (3) at (2.6,0){};
								\node[draw,circle,thick,scale=1.25,label=above:{1}, fill=black] (4) at (0,.5){};
				\draw[thick] (1) to (2);\draw[thick] (2) to (4);
				\draw[thick] (1.5,0.09) --++ (.9,0);
				\draw[thick] (1.5,-0.09) --++ (.9,0);
				\draw[thick]
					(1.9,0) --++ (60:.25)
					(1.9,0) --++ (-60:.25);
			\end{tikzpicture}\end{array}
		$} & 
		\scalebox{.8}{$\begin{array}{c}\begin{tikzpicture}
				\node[draw,circle,thick,scale=1.25,label=30:{2}] (0) at (0,0){};
				\node[draw,circle,thick,scale=1.25,label=below:{1}, fill=black] (1) at (-1,0){};
				\node[draw,circle,thick,scale=1.25,label=below:{1}] (2) at (1,0){};
				\node[draw,circle,thick,scale=1.25,label=above:{1}] (3) at (90:1){};      								\node[draw,circle,thick,scale=1.25,label=below:{1}] (4) at (90:-1){};
				\draw[thick] (0) to (1);
				\draw[thick] (0) to (2);
				\draw[thick] (0) to (3);
				\draw[thick] (0) to (4);
								\draw[<->,>=stealth',semithick,dashed]  (1.1,0.3) arc (25:65:1.7cm);
			\end{tikzpicture}\end{array}
		$}

				\\\hline

			\begin{array}{c}
			\text{I}^{*\text{ns}}_{0}
			\\
			\widetilde{\text{G}}_2^t
\end{array}
 &			\scalebox{1}{$\begin{array}{c}\begin{tikzpicture}
				\node[draw,circle,thick,scale=1.25,label=above:{1}, fill=black] (1) at (0,0){};
				\node[draw,circle,thick,scale=1.25,label=above:{2}] (2) at (1.3,0){};
				\node[draw,circle,thick,scale=1.25,label=above:{1}] (3) at (2.6,0){};
				\draw[thick] (1) to (2);
				\draw[thick] (1.5,0.09) --++ (.9,0);
				\draw[thick] (1.5,-0.09) --++ (.9,0);
				\draw[thick] (1.5,0) --++ (.9,0);
				\draw[thick]
					(1.9,0) --++ (60:.25)
					(1.9,0) --++ (-60:.25);
			\end{tikzpicture}\end{array}
		$} & 
		\scalebox{.8}{$\begin{array}{c}\begin{tikzpicture}
				\node[draw,circle,thick,scale=1.25,label=30:{2}] (0) at (0,0){};
				\node[draw,circle,thick,scale=1.25,label=above:{1}, fill=black] (1) at (-1,0){};
				\node[draw,circle,thick,scale=1.25,label=right:{1}] (2) at (1,0){};
				\node[draw,circle,thick,scale=1.25,label=above:{1}] (3) at (90:1){};
								\node[draw,circle,thick,scale=1.25,label=below:{1}] (4) at (90:-1){};
				\draw[thick] (0) to (1);
				\draw[thick] (0) to (2);
				\draw[thick] (0) to (3);
				\draw[thick] (0) to (4);
				\draw[<->,>=stealth',semithick,dashed]  (1.1,0.3) arc (25:65:1.7cm);
				\draw[<->,>=stealth',semithick,dashed]  (1.1,-0.3) arc (-25:-65:1.7cm);
			\end{tikzpicture}\end{array}
		$}

				\\\hline

				\end{array}$
	\end{center}
	\caption{
{Dual graphs for singular fibers elliptic fibrations with non-geometrically irreducible fiber components.}  \label{Table:DualGraph}
}
\end{table}

\clearpage

\subsection{Crepant resolution}\label{Sec:Blowup}

A Weierstrass model over a base $B$ is defined in the ambient space 
\begin{equation}
X_0=\mathbb{P}(\mathscr{O}_B\oplus \mathscr{L}^{\otimes 2}\oplus \mathscr{L}^{\otimes 3}),
\end{equation}  
for a choice of a line bundle $\mathscr{L}$ over $B$ called the fundamental line bundle of the Weierstrass model. 
  Each crepant resolution is an embedded resolution defined by a sequence of blowups with smooth centers. 
We denote the blowup $X_{i+1}\to X_i$ along the ideal $(f_1,f_2,\ldots,f_n)$ with exceptional divisor $E$ as:
$$\begin{tikzcd}[column sep=2.4cm]X_i \arrow[leftarrow]{r} {\displaystyle (f_1,\ldots, f_n|E)}  & X_{i+1}\end{tikzcd},$$
where $X_0$ is the projective bundle in which the Weierstrass model is defined. 

 An important property of generic USp($2n$)-model is that they have at most one crepant resolution. 
 A crepant  resolution of the USp($4$)-model  is given by the following sequence of blowups of $X_0$ with smooth centers: 
\begin{align}
		\begin{array}{c}
			\begin{tikzpicture}
				\node(E0) at (0,0) {$X_0$};
				\node(E1) at (4,0) {$X_1$};
				\node(E2) at (8,0) {$X_2$};
				\draw[big arrow] (E1) -- node[above,midway]{ $(x,y,s|e_1)$}  (E0);
				\draw[big arrow] (E2) -- node[above,midway]{ $(x,y,e_1|e_2)$}  (E1);
			\end{tikzpicture}
		\end{array}
		\end{align}
		where we abuse notation and write $(x,y,s)$ for the proper transform of $(x,y,s)$ after the first blowup and 
		we write ($x,y,e_1)$ for the proper transform of ($x,y,e_1)$ after the second blowup. 
 The proper transform of the elliptic fibration is:
\begin{equation}
Y: \quad zy^2=e_1 e_2^2 x^3 + a_2 x^2 z+ a_4 e_1 s^2 xz^2 + a_6 e_1^2 s^4z^3.
\end{equation}
The relative projective coordinates are:
\begin{equation}\label{proj.cord3}
[e_1 e_2^2 x:e_1 e_2^2y:z] [e_2x:e_2y:s][x:y:e_1].
\end{equation}
Working in patches, it is easy to check that the two blowups produce a smooth variety by the Jacobian criterion. 
The resolution is crepant because at each blowup we factor two exceptional divisors after blowing up an regular sequence of length three.

\subsection{Fiber degenerations and the geometry of the fibral divisors}
After the two blowups defining the crepant resolution, the total transform of $S=V(s)$ is 
\begin{equation}
f^* s= s e_1 e_2 
\end{equation}
The corresponding fibral divisors are:
\begin{align}
\begin{cases}
D_0: &\quad  s=zy^2-x^2(e_1 e_2^2 x - a_2 z)=0,\\
D_1: &\quad e_1= y^2-a_2 x^2 =0,\\
D_2: &\quad e_2=y^2- a_2 x^2 -a_4 e_1 s^2 xz -a_6 e_1^2 s^4z^2=0.
\end{cases}
\end{align}
We denote respectively  by  $C_0$, $C_1$, and $C_2$ the generic fiber of the fibral divisors $D_0$, $D_1$, and $D_2$. 
By looking carefully at their defining equations, we see that the fibers of $D_0$, $D_1$, and $D_2$ are  always rational curves and never jump in dimension. Thus, the elliptic fibration is equidimensional and this implies that the fibration is also  flat since we work over an algebraically closed field. 

 We will now study carefully the geometry of the generic curves $C_a$ and the fibral divisors $D_a$ ($a=0,1,2$). 
The curve $C_0$ is a rational curve that can be parametrized by $y/x$. More precisely, $C_0$ is the normalization of a nodal curve. The divisor $D_0$ is a $\mathbb{P}^1$-bundle over $S$ since $C_0$ does not degenerate. Thus, we have
	\begin{align}
		D_0\cong \mathbb{P}_S[\mathscr{O}_S\oplus \mathscr{L}].
	\end{align}

 To understand the geometry of $D_1$, we use the Stein factorization as in \cite{G2,F4}. 
 The fibral divisor $D_1$ is a double cover of a variety  $\widetilde{D}_1$  with branch locus $V(a_2)$  such that   $\widetilde{D}_1$ has connected fibers over $S$. More precisely,  $\widetilde{D}_1$ is a  $\mathbb{P}^1$-bundle $ \mathbb{P}_S[\mathscr{L}\oplus \mathscr{O}_S]$ over $S$:
\begin{align}
 D_1 \overset{2:1}{\longrightarrow}   \widetilde{D}_1\cong \mathbb{P}_S[\mathscr{L}\oplus \mathscr{O}_S]\longrightarrow S.
\end{align}
The fiber  $C_1$ is given by a conic that is geometrically degenerated.  Geometrically, the curve $C_1$ is composed of two distinct rational curves that coincide into a double curve over $V(a_2)\cap S$. 
Each of these rational curves defines a $\mathbb{P}^1$-bundle isomorphic to $ \mathbb{P}_S[\mathscr{L}\oplus \mathscr{O}_S]$. 
The divisor $D_1$ is obtained by gluing two such bundles along their fibers over $V(a_2)\cap S$.
In the total space of the line bundle $\mathscr{L}$ over $S$, $D_1$ can be written as the double cover $V(t^2 -a_2)$ where $t=y/x$ is a section of $\mathscr{L}$ restricted on $S$.

The fibral divisor $D_2$ is a  conic bundle with matrix (with respect to the variables $(y, x, e_1 s^2 z)$):
\begin{equation}
M=
\begin{pmatrix}
1 & 0 & 0 \\
0 & -a_2 & -\frac{a_{4,2}}{2},\\
0 & -\frac{a_{4,2}}{2} & -a_{6,4}.
\end{pmatrix},
\end{equation}
whose discriminant locus is: 
\begin{equation}
V(\text{det}\, M) = V(a_4^2-4 a_2 a_6).
\end{equation}
The degeneration of the generic fiber $C_2$  of $D_2$ is controlled by the branch locus of $D_2$ and the rank of the matrix $M$. The curve $C_2$ is a smooth conic if $M$ has rank 3, $C_2$ degenerates into two lines meeting transversally at a point when 
the rank of $M$ is 2, and finally, $C_2$ degenerates into a double line when the rank of $M$ is one:
\begin{equation}
\text{rk}\  M=\begin{cases}
 1 \quad \iff \quad a_2=a_{4,2}=a_{6,4}=0, \\
 2 \quad\iff \quad  a_{4,2}^2-4 a_2 a_{6,4}=0 \quad \text{and} \quad  (a_2,a_{4,2},a_{6,4})\neq (0,0,0), \\
 3 \quad \iff \quad  a_{4,2}^2-4 a_2 a_{6,4}\neq 0.
\end{cases}
\end{equation}
These are also the irreducible components of the intersection of the two components of the discriminant locus. 
It follows that the generic fiber of $D_2$ can only degenerate over $V(S)\cap V(a_2)$, $V(S)\cap V( a_4^2-4 a_2 a_6)$, and 
$V(S)\cap V( a_2, a_4)$. 
The degeneration of $C_1$ and $C_2$ are described in section \ref{sec:deg}, see also Figure
\ref{Fig.V1}, Figure \ref{Fig.V3}, and Figure \ref{Fig.V2}.

\section{Intersection numbers, topological and characteristic invariants }
In this section we collect several important data from the intersection ring of a USp($4$)-model. 
First, we compute the intersection numbers of vertical rational curves with the fibral divisors;
these numbers are understood as the Dynkin labels of the weights of certain representations. 
We then compute the generating function for the Euler characteristic of a USp($4$)-model.  

\subsection{Representations from intersection numbers}\label{sec:deg}
Over $V(S)\cap V(a_2)$, the component 
$C_1$ of the fiber over the generic point of V(S) 
collapses into a double rational curve ($C_1\to  2 C_1'$):
\begin{align}
V(a_2)\Longrightarrow
\begin{cases}
C_0\to C_0 \\
 C_1\to 2 C'_1\\
 C_2 \to C_2 
 \end{cases}
\end{align}
The fiber irreducible components of the fiber over the generic point of 
 $V(S)\cap V(a_2)$ 
are 
\begin{align}
\begin{cases}
C_0: &\quad  s=zy^2-e_1 e_2^2 x^3 - a_2 x^2z=0\\
C'_1 :&\quad e_1= y=0\\
C_2 :&\quad e_2=y^2-a_4 e_1 s^2 x z -a_6 e_1^2 s^4z^2=0
\end{cases}
\end{align}

The weight  of the new curve $C_1'$ with respect to $(D_0, D_1, D_2)$ is half of the weight of $C_1$, namely $[0,-1,2]$. 
Ignoring the contribution from $D_0$, this weight  is \boxed{-1, 2} and its Weyl orbit consists of the non-zero weights of the representation $\mathbf{5}$ of C$_2$.  
The representation  $\mathbf{5}$ of C$_2$ is quasi-minuscule and real.
\begin{equation}
a_2=0\Longrightarrow [0,-1,2]\Longrightarrow \boxed{-1,2}\Longrightarrow \mathbf{5}.
\end{equation}

Over $V(S)\cap V(a_4^2-4 a_2 a_6)$, the conic defining $C_2$ degenerates to two intersecting lines.  
\begin{align}
V(a_4^2-4 a_2 a_6)\Longrightarrow
\begin{cases}
C_0\to C_0 \\
 C_1\to  C_1\\
 C_2 \to C_2^{-}+C_2^{+} 
 \end{cases}
\end{align}
But the two lines require to take the square root of $a_2$. 
That means that the two curves are non-split  just as $C_1$ and the singular fiber is of type  I$_5^{\text{ns}}$. 
Both geometric irreducible components of $C_2$ have weight $[0,1,-1]$ with respect to $(D_0,D_1,D_2)$. 

The Weyl orbit of $[1,-1]$ is the set of weights of the representation $\mathbf{4}$ of C$_2$. 
The representation $\mathbf{4}$ is  minuscule and 
  pseudo-real.
\begin{equation}
a_4^2-4 a_2 a_6=0\Longrightarrow 
[0,1,-1]\Longrightarrow \boxed{1,-1}\Longrightarrow \mathbf{4}.
\end{equation}
Finally, over $V(S)\cap V(a_2, a_4)$, $C_1$ collapses to a double curve while the conic $C_2$ splits geometrically into two curves. 
Again the components of $C_2$ are only visible after a field extension, this time the field extension requires taking the square root of $a_6$.

\begin{rem}
If the base is a surface, the fibers over codimension two points are always geometric since we are over closed points. In particular, over $V(a_4^2-4 a_2 a_6)$,  the fiber is of type I$_5$ and we can compute the weights of $C_{2}^\pm$ individually. 
They produce the representation $\mathbf{4}$. 
If the base is three dimensional or higher, over 
  $V(a_4^2-4 a_2 a_6)$, the fiber is of type I$_5^{\text{ns}}$. While $C_2$ geometrically splits into two curves $C_2^\pm$, these two curves are non-split and do not produce new weights.
  In particular, we do not get the representation $\mathbf{4}$ when the base has dimension $3$ or higher. 
\end{rem}

We have 
	\begin{equation}
\varphi_* \big(D_a\cdot C_b \Big)=
\begin{tabular}{r}
$\begin{matrix}
 C_0 & C_1  & C_2
\end{matrix}\   \  $\\
$\begin{matrix}
D_0 \\
D_1\\
D_2
\end{matrix}
\begin{bmatrix}
-2 & \   \  2 &\  \  0 \\
\   \  2 & -4 &\  \   2\\
\  \   0&\   \    2& -2
\end{bmatrix} 
$
\end{tabular},
	\end{equation}
	which  corresponds to the affine twisted Dynkin diagram $\tilde{C}_2^t$.
	The weights of all the curves are 
	\begin{equation}
\begin{array}{|r|r|r||r|r|r|r|r|}
\cline{2-8}
\multicolumn{1}{c|}{} & C_0& C_1 & C_2& C_1'& C_1^\pm & C_2^\pm & C_2' \\
\cline{1-8}
D_0 &-2 & 2 &0 & 1& 1& 0&0\\ 
D_1&2    &  -4& 2& -2&-2 & 1&1\\ 
 D_2 & 0   & 2 & -2& 1&1 &-1 &-1\\ 
   \cline{1-8}
\end{array}
\end{equation}
	
	\subsection{Euler characteristic and Hodge numbers}\label{Sec:Euler}
\begin{thm}\label{Thm:Euler}
The Euler characteristic of a \textup{USp($4$)} model over a base of dimension $d$  is the component of degree $d$ of the following rational function in the Chow ring of the base:
\begin{equation}
\chi(Y)=4\int_B   \frac{3 L (1 + 2 L) + 6 L (1 + 2 L) S - (5 + 8 L) S^2 }{(1 + 2 L) (1 + 6 L - 4 S) (1 + S)}
 c(TB).
\end{equation}
where  $c(TB)$ is the Chern class of the tangent bundle of $B$, $L=c_1(\mathscr{L})$, and $S$ is the class of the divisor $V(s)$.
\end{thm}

\begin{proof}
See \cite[Table 7]{Euler}.
\end{proof}
\begin{rem}
This theorem generalized results from [Grassi-Morrison] for Calabi-Yau threefolds while removing most of the conditions. 
In particular, we do not assume the Calabi-Yau condition, and we do not restrict the dimension of the base. 
\end{rem}
\begin{lem}\label{Lem:Hodge}
If the USp($4$)-model $Y\to B$ is a threefold and $c_1=c_1(TB)$ is the first Chern class of the base, we have
\begin{equation}
\chi(Y)=
4 (3 c_1 L - 18 L^2 + 15 L S - 5 S^2).
\end{equation}
If $Y$ is a Calabi-Yau threefold then the adjuction formula gives $L=c_1(TB)$ and:
\begin{equation}
\chi(Y)=-20 (3 c_1^2 - 3 c_1 S + S^2).
\end{equation}
For a smooth base with canonical class $K$ and a smooth curve  $S$ of genus $g$ and self-intersection $S^2$:
\begin{equation}
\chi(Y)=-60 K^2 +40 S^2 +120(1-g)= -60 K^2 + 60 n  (1-g) +4 (7- n)n S^2, 
\end{equation}
A USp(4) model that is a Calabi-Yau threefold has Hodge numbers: 
      \begin{align}
   h^{1,1}(Y) & =13-K^2, \quad h^{1,2}(Y)=-47 + 60 g + 29 K^2 - 20 S^2      \end{align}
\end{lem}

\begin{proof}
The Calabi--Yau condition is $L=-K$ and the Hodge numbers follows from the Shioda--Tate--Wazir formula  \cite{Wazir} as used in  \cite[Theorem 4.9]{Euler}.
\end{proof}

\begin{conjec}
For a USp($2n$)-model, we have 
\begin{equation}
\chi(Y)=2\int_B   \frac{6 L (1 + 2 L) + 3 L (1 + 2 L) (2 + n) S - 
 n (1 + 4 L + 2 n + 2 L n) S^2 }{(1 + 2 L) (1 + 6 L - 2n S) (1 + S)}
 c(TB).
\end{equation}
In the Calabi--Yau threefold case, we have 
\begin{equation}
h^{1,1}(Y)  = 11+n -K^2, \quad h^{1,2}(Y)= (29 -30g)n  -11  - 29 K^2  + 2n(7-n)  S^2.
\end{equation}
\end{conjec}
The conjecture has been checked up to $n=11$.  A complete proof will be presented elsewhere.

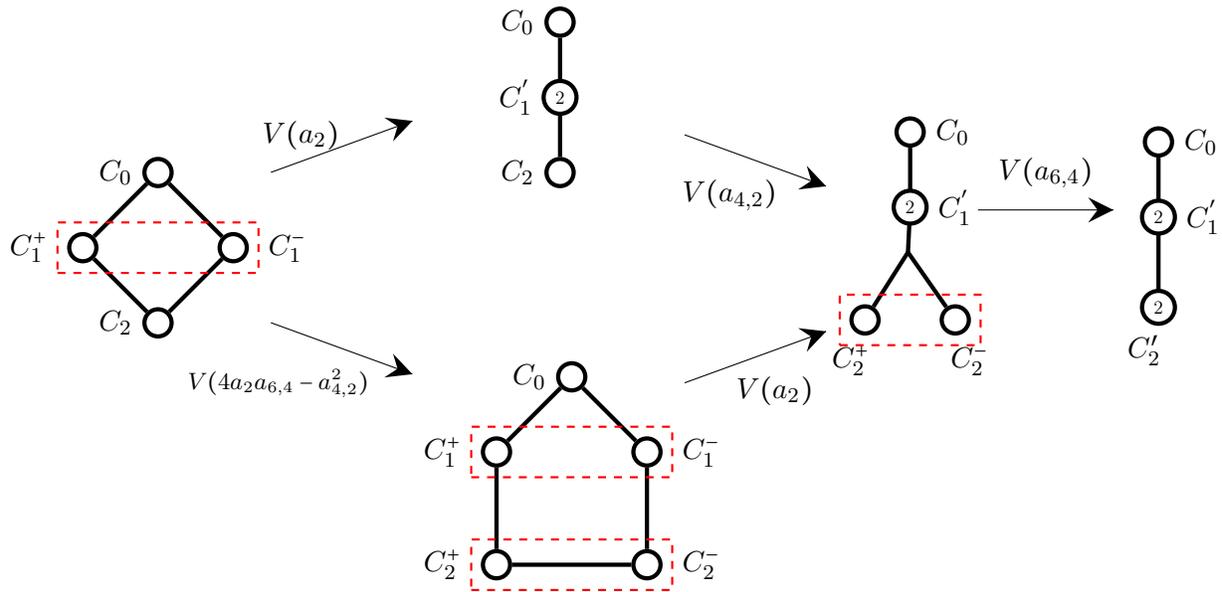
\begin{figure}
\begin{center}
\scalebox{1}{
\begin{tikzpicture}[scale=1]
					\node at (3,2) 
						{
							\begin{tikzpicture}
								\node[draw,ultra thick,circle,xshift=200,label=left:{$C_0$}](C0') at (0,0) {};
								\node[draw,ultra thick,circle,xshift=200,scale=.65,label=left:{$C_1^{'}$}](C13) at (0,-1) {2};
								\node[draw,ultra thick,circle,xshift=200,label=left:{$C_2$}](C2) at (0,-2) {};
								\draw[ultra thick] (C0') -- (C13) -- (C2);
							\end{tikzpicture}
						};
						\tikzset{myptr/.style={decoration={markings,mark=at position 1 with %
    {\arrow[scale=3,>=stealth]{>}}},postaction={decorate}}}
\draw [myptr] (-.5,1) -- +(20:2);
\draw [myptr] (-.5,-1) -- +(-20:2);
\draw [myptr] (5,1.5) -- +(-20:2);
\draw [myptr] (5,-1.8) -- +(20:2);
\draw [myptr] (8.9,.5) -- +(1.8,0);
\node at (-.1,1.5) {$V(a_2)$};
\node at (5.6,.7) {$V(a_{4,2})$};
\node at (6.2,-1.9) {$V(a_{2})$};
\node at (-.4,-1.8) {\footnotesize $V(4a_2 a_{6,4} -a_{4,2}^2)$};
\node at (9.8,1) {$V(a_{6,4})$};

					\node at (-2,0) 
						{
							\begin{tikzpicture}[]
								\node[draw,ultra thick,circle,label=left:{$C_0$}](C0) at (0,0) {};
								\node[draw,ultra thick,circle,label=left:{$C_1^+\  $}](C1+) at (-1,-1) {};
								\node[draw,ultra thick,circle,label=right:{$\ C_1^-$}](C1-) at (1,-1) {};
								\node[draw,ultra thick,circle,label=left:{$C_2$}](C2) at (0,-2) {};;
								\draw[ultra thick] (C0) -- (C1+) -- (C2) -- (C1-) -- (C0);
								\node[draw,color=red,thick,  dashed,fit=(C1+) (C1-)] {};
							\end{tikzpicture}
						};
						
							\node at (3.5,-3) 
						{
							\begin{tikzpicture}[]
								\node[draw,ultra thick,circle,label=left:{$C_0$}](C0) at (0,0) {};
								\node[draw,ultra thick,circle,label=left:{$C_1^+\  $}](C1+) at (-1,-1) {};
								\node[draw,ultra thick,circle,label=right:{$ \ C_1^-$}](C1-) at (1,-1) {};
								\node[draw,ultra thick,circle,label=left:{$C_2^+\   $}](C21) at (-1,-2.5) {};
							  \node[draw,ultra thick,circle,label=right:{$ \   C_2^-$}](C22) at (1,-2.5) {};
								\draw[ultra thick] (C0) -- (C1+) -- (C21)--(C22) -- (C1-) -- (C0);
								\node[draw,color=red,thick,  dashed,fit=(C21) (C22)] {};
								
								\node[draw,color=red,thick,  dashed,fit=(C1+) (C1-)] {};
							\end{tikzpicture}
						};
						\node at (8,0) 
						{
							\begin{tikzpicture}
								\node[draw,ultra thick,circle,xshift=200,label=right:{$C_0$}](C0') at (0,0) {};
								\node[draw,ultra thick,circle,xshift=200,scale=.65,label=right:{$C_1^{'}$}](C13) at (0,-1) {2};
					
								\node[draw,ultra thick,circle,xshift=200,label=below:{$C_2^+\quad$}](C2+) at (-.6,-2.5) {};
								\node[draw,ultra thick,circle,xshift=200,label=below:{$\quad C_2^-$}](C2-) at (.6,-2.5) {};
								\draw[ultra thick] (C0') -- (C13) -- (7,-1.6);
								\draw[ultra thick] (C2-) -- (7,-1.6) -- (C2+);
								
								\node[draw,color=red,thick,  dashed,fit=(C2+) (C2-)] {};
							\end{tikzpicture}
						};

						\node at (11.5,0) 
						{
							\begin{tikzpicture}
								\node[draw,ultra thick,circle,xshift=200,label=right:{$C_0$}](C0') at (0,0) {};
								\node[draw,ultra thick,circle,xshift=200,scale=.65,label=right:{$C_1^{'}$}](C13) at (0,-1) {2};
					
								\node[draw,ultra thick,circle,xshift=200,scale=.65,label=below:{$C'_2\quad$}](C2) at (0,-2.2) {2};
								\draw[ultra thick] (C0') -- (C13) -- (C2);
								
							\end{tikzpicture}
						};
						\end{tikzpicture}
						}
					\end{center}
						\caption{Fiber degeneration of  a  USp(4)-model with $v(a_4)=2$ and $v(a_6)=4$. 	\label{Fig.V1}} 
						\end{figure}

	\begin{figure}
	\begin{center}
\begin{tikzpicture}[scale=1]
					\node at (3,2) 
						{
							\begin{tikzpicture}
								\node[draw,ultra thick,circle,xshift=200,label=left:{$C_0$}](C0') at (0,0) {};
								\node[draw,ultra thick,circle,xshift=200,scale=.65,label=left:{$C_1^{'}$}](C13) at (0,-1) {2};
								\node[draw,ultra thick,circle,xshift=200,label=left:{$C_2$}](C2) at (0,-2) {};
								\draw[ultra thick] (C0') -- (C13) -- (C2);
							\end{tikzpicture}
						};
						\tikzset{myptr/.style={decoration={markings,mark=at position 1 with %
    {\arrow[scale=3,>=stealth]{>}}},postaction={decorate}}}
\draw [myptr] (-.5,1) -- +(20:2);
\draw [myptr] (-.5,-1) -- +(-20:2);
\draw [myptr] (5,1.5) -- +(-20:2);
\draw [myptr] (5,-1.8) -- +(20:2);
\node at (-.1,1.5) {$V(a_2)$};
\node at (5.6,.7) {$V(a_{4,2})$};
\node at (6.2,-1.9) {$V(a_{2})$};
\node at (-.1,-1.7) {$V(a_{4,2})$};

					\node at (-2,0) 
						{
							\begin{tikzpicture}[]
								\node[draw,ultra thick,circle,label=left:{$C_0$}](C0) at (0,0) {};
								\node[draw,ultra thick,circle,label=left:{$C_1^+\  $}](C1+) at (-1,-1) {};
								\node[draw,ultra thick,circle,label=right:{$\ C_1^-$}](C1-) at (1,-1) {};
								\node[draw,ultra thick,circle,label=left:{$C_2$}](C2) at (0,-2) {};;
								\draw[ultra thick] (C0) -- (C1+) -- (C2) -- (C1-) -- (C0);
								\node[draw,color=red,thick,  dashed,fit=(C1+) (C1-)] {};
							\end{tikzpicture}
						};
						
							\node at (3.5,-3) 
						{
							\begin{tikzpicture}[]
								\node[draw,ultra thick,circle,label=left:{$C_0$}](C0) at (0,0) {};
								\node[draw,ultra thick,circle,label=left:{$C_1^+\  $}](C1+) at (-1,-1) {};
								\node[draw,ultra thick,circle,label=right:{$ \ C_1^-$}](C1-) at (1,-1) {};
								\node[draw,ultra thick,circle,label=left:{$C_2^+\   $}](C21) at (-1,-2.5) {};
							  \node[draw,ultra thick,circle,label=right:{$ \   C_2^-$}](C22) at (1,-2.5) {};
								\draw[ultra thick] (C0) -- (C1+) -- (C21)--(C22) -- (C1-) -- (C0);
								\node[draw,color=red,thick,  dashed,fit=(C21) (C22)] {};
								
								\node[draw,color=red,thick,  dashed,fit=(C1+) (C1-)] {};
							\end{tikzpicture}
						};
						\node at (8,0) 
						{
							\begin{tikzpicture}
								\node[draw,ultra thick,circle,xshift=200,label=right:{$C_0$}](C0') at (0,0) {};
								\node[draw,ultra thick,circle,xshift=200,scale=.65,label=right:{$C_1^{'}$}](C13) at (0,-1) {2};
					
								\node[draw,ultra thick,circle,xshift=200,scale=.65,label=below:{$C'_2\quad$}](C2) at (0,-2.2) {2};
								\draw[ultra thick] (C0') -- (C13) -- (C2);
								
							\end{tikzpicture}
						};

						\end{tikzpicture}
						\caption{Fiber degeneration of  a  USp(4)-model with $v(a_4)=2$ and $v(a_6)\geq 5$. 	
						\label{Fig.V3}
											} 
						\end{center}
						\end{figure}

	\begin{figure}
	\begin{center}
	\scalebox{1}{
\begin{tikzpicture}[scale=1]
					\node at (3,2) 
						{

\begin{tikzpicture}
								\node[draw,ultra thick,circle,xshift=200,label=right:{$C_0$}](C0') at (0,0) {};
								\node[draw,ultra thick,circle,xshift=200,scale=.65,label=right:{$C_1^{'}$}](C13) at (0,-1) {2};
					
								\node[draw,ultra thick,circle,xshift=200,label=below:{$C_2^+\quad$}](C2+) at (-.6,-2.5) {};
								\node[draw,ultra thick,circle,xshift=200,label=below:{$\quad C_2^-$}](C2-) at (.6,-2.5) {};
								\draw[ultra thick] (C0') -- (C13) -- (7,-1.6);
								\draw[ultra thick] (C2-) -- (7,-1.6) -- (C2+);
								
								\node[draw,color=red,thick,  dashed,fit=(C2+) (C2-)] {};
							\end{tikzpicture}
						};
						\tikzset{myptr/.style={decoration={markings,mark=at position 1 with %
    {\arrow[scale=3,>=stealth]{>}}},postaction={decorate}}}
\draw [myptr] (-.5,1) -- +(20:2);
\draw [myptr] (-.5,-1) -- +(-20:2);
\draw [myptr] (5,1.5) -- +(-20:2);
\draw [myptr] (5,-1.8) -- +(20:2);
\node at (-.1,1.5) {$V(a_2)$};
\node at (5.6,.7) {$V(a_{6,4})$};
\node at (6.2,-1.9) {$V(a_{2})$};
\node at (-.1,-1.7) {$V(a_{6,4})$};

					\node at (-2,0) 
						{
							\begin{tikzpicture}[]
								\node[draw,ultra thick,circle,label=left:{$C_0$}](C0) at (0,0) {};
								\node[draw,ultra thick,circle,label=left:{$C_1^+\  $}](C1+) at (-1,-1) {};
								\node[draw,ultra thick,circle,label=right:{$\ C_1^-$}](C1-) at (1,-1) {};
								\node[draw,ultra thick,circle,label=left:{$C_2$}](C2) at (0,-2) {};;
								\draw[ultra thick] (C0) -- (C1+) -- (C2) -- (C1-) -- (C0);
								\node[draw,color=red,thick,  dashed,fit=(C1+) (C1-)] {};
							\end{tikzpicture}
						};
						
							\node at (3.5,-3) 
						{
							\begin{tikzpicture}[]
								\node[draw,ultra thick,circle,label=left:{$C_0$}](C0) at (0,0) {};
								\node[draw,ultra thick,circle,label=left:{$C_1^+\  $}](C1+) at (-1,-1) {};
								\node[draw,ultra thick,circle,label=right:{$ \ C_1^-$}](C1-) at (1,-1) {};
								\node[draw,ultra thick,circle,label=left:{$C_2^+\   $}](C21) at (-1,-2.5) {};
							  \node[draw,ultra thick,circle,label=right:{$ \   C_2^-$}](C22) at (1,-2.5) {};
								\draw[ultra thick] (C0) -- (C1+) -- (C21)--(C22) -- (C1-) -- (C0);
								\node[draw,color=red,thick,  dashed,fit=(C21) (C22)] {};
								
								\node[draw,color=red,thick,  dashed,fit=(C1+) (C1-)] {};
							\end{tikzpicture}
						};
						\node at (8,0) 
						{

							\begin{tikzpicture}
								\node[draw,ultra thick,circle,xshift=200,label=right:{$C_0$}](C0') at (0,0) {};
								\node[draw,ultra thick,circle,xshift=200,scale=.65,label=right:{$C_1^{'}$}](C13) at (0,-1) {2};
					
								\node[draw,ultra thick,circle,xshift=200,scale=.65,label=below:{$C'_2\quad$}](C2) at (0,-2.2) {2};
								\draw[ultra thick] (C0') -- (C13) -- (C2);
								
							\end{tikzpicture}
						};

						\end{tikzpicture}}
						\caption{Fiber degeneration of  a  USp(4)-model with $v(a_4)\geq 3$ and $v(a_6)=4$. \label{Fig.V2}}
						\end{center}
						\end{figure}

	\subsection{Triple intersection numbers}

\begin{thm}
The triple intersection polynomial of the fibral divisors is $6\mathcal{F}=(D_0\alpha_0+D_1\alpha_1+D_2\alpha_2)^3$ with 						
\begin{equation}				
\begin{aligned}
6 {\mathcal F}&=
-4 S (-L + S) \alpha_0^3 - 6 (2 L - S) S \alpha_0^2 \alpha_1 + 
 12 L S \alpha_0 \alpha_1^2  \\
 & \quad  - 8 S^2 \alpha_1^3 -
 12 (L - S) S \alpha_1^2 \alpha_2 - 
 6 S (-2 L + S) \alpha_1 \alpha_2^2 - 4 L S \alpha_2^3
\end{aligned}			
\end{equation}

\end{thm}
\begin{proof}
The classes of the three divisors are
$$
[D_0]=f_2^* f_1^*\pi^* S-f_2^* E_1, \quad [D_1]=f_2^*E_1-E_2, \quad [D_2]=E_2
$$
where $E_i$ is the class of the exceptional divisor of the $i$th blowup $f_i$. The class of the Weierstrass equation before the blowup is $(3H+6L)$ and after the two blowups it becomes 
$(3H+6L-2E_1-2E_2)$. 
The rest follows from the pushforward formula
$$ 6\mathcal{F}=\int_B \pi_* f_{1*} f_{2*} \Big( ([D_0]\alpha_0 + [D_1]\alpha_1 +[D_2]\alpha_2)^3 (3f_2^*f_1^* H +6 f_2^*f_1^*\pi^*L-2 f_2^* E_1-2E_2)\Big).$$
To perform the pushforward, we use the following data from the centers of the blowups: 
\begin{equation}
\begin{array}{lll}
Z_1^{(1)} = H+3\pi^* L  \quad &,\quad Z_2^{(1)} =H+2\pi^* L  \quad &,\quad Z_3^{(1)}=\pi^* S,\\
Z_1^{(2)}=f_1^* H+3f_1^* \pi^*L-E_1 \quad &, \quad  Z_2^{(2)}=f_1^* H+2f_1^*\pi^*L-E_1 \quad&, \quad Z_3^{(3)}=f_1^* E_1.
\end{array}
\end{equation}
The triple intersection numbers involve at most quartic monomials. 
The pushforward after a blowup of the complete intersection ($Z_1, Z_2,Z_3)$ with exceptional divisor $E$ is \cite{Euler}
\begin{equation}\label{eq:Int1}
f_*E = 0,\quad  f_*E^2= 0, \quad f_*E^3= Z_1Z_2Z_3, \quad f_*E^4= (Z_1 + Z_2+Z_3)Z_1Z_2Z_3,
\end{equation}
and the pull-push  formula for any proper map $g$: 
\begin{equation}\label{eq:Int2}
g_*(\alpha\cdot  g^*   \beta)=g_* \alpha\cdot \beta,
\end{equation}
and 
\begin{equation}  \label{eq:Int3}
\pi_* H=0, \quad \pi_* H^2 =1, \quad \pi_* H^3=-5L, \quad \pi^* H^4=19 L^2.
\end{equation}
 The equations 
  \eqref{eq:Int1}, \eqref{eq:Int2}, and \eqref{eq:Int3}
 are enough to compute the triple intersection numbers. 

\end{proof}

\begin{lem}
In the case where the resolved elliptic fibration $Y$ is a Calabi-Yau threefold, with $S$ a smooth curve of genus $g$: 
					\begin{equation}				
\begin{aligned}
6 {\mathcal F}&=
8(1-g) \alpha_0^3 - 6 (4  -4 g+S^2) \alpha_0^2 \alpha_1 + 
 12(2-2g +S^2) \alpha_0 \alpha_1^2  \\
 & \quad  - 8 S^2 \alpha_1^3 -
 24 (1 - g)  \alpha_1^2 \alpha_2 - 
 6 (-4+4 g-S^2) \alpha_1 \alpha_2^2 - 4 (2-2g+S^2)  \alpha_2^3.
\end{aligned}			
\end{equation}
\end{lem}

\begin{proof}
The Calabi--Yau condition gives $L=-K$ and the adjunction gives $2-2g=-K\cdot S-S^2$, which we can use to eliminate appearances of $K\cdot S=-L\cdot S$ . 
\end{proof}

\section{Compactifications of M-theory to 5D on a USp($4$)-model}

M-theory compactified on a Calabi-Yau threefold results in a five-dimensional supersymmetric gauge theory with eight supercharges that we denote 
${\cal N}=1$ five-dimensional supergravity \cite{Cadavid:1995bk,Ferrara:1996wv,IMS} and its matter content consists of  a gravitational multiplet, n$_T$ tensor multiplets, n$_V$ vector multiplets, and n$_H$ hypermultiplets. 
In five dimensional spacetime, a massless tensor  multiplets is dual to a massless vector. 
We assume that all tensors are massless and are dualized to vectors.  
Each vector multiplet contain a real scalar field $\phi$ and each hypermultiplet contains four real fields forming  a quaternion. 
 The kinetic terms of all the vector multiplets and the graviphoton as well as the Chern-Simons terms are determined by a real function of the scalar fields called the prepotential.

 In the Coulomb branch of an $\mathcal{N}=1$ supergravity theory in five dimension, the Lie group is completely broken to a Cartan subgroup.  This implies that the charge of a hypermultiplet is simply a  
 weight of the representation under which it transforms \cite{IMS}.  
The Intrilligator-Morrison-Seiberg (IMS) prepotential is the quantum contribution to the prepotential of a five-dimensional gauge theory after integrating out all massive fields. 

Let $\phi$ be in the Cartan subalgebra of a Lie algebra $\mathfrak{g}$. We denote by  $\mathbf{R}_i$ the representations under which the hypermultiplets transform. 
The  weights are in the dual space of the Cartan subalgebra.  We denote the evaluation of a  weight on a coroot vector $\phi$ as a scalar product $\langle \mu,\phi \rangle$.  Denoting the roots by $\alpha$ and the weights of $\mathbf{R}_i$ by $\varpi$ we have \cite{IMS}
\begin{align}
6\mathscr{F}_{\text{IMS}}(\phi) =&\frac{1}{2} \left(
\sum_{\alpha} |\langle \alpha, \phi \rangle|^3-\sum_{i} \sum_{\varpi\in \mathbf{R}_i} n_{\mathbf{R}_i} |\langle \varpi, \phi\rangle|^3 
\right).\label{Eq:IMS}
\end{align}
For all simple groups with the exception of SU$(N)$ with $N\geq 3$, this is the full purely cubic sector of the prepotential as there are no non-trivial third order Casimir invariants.

\subsection{5D prepotential and counting charged hypermultiplets}
\begin{thm}
The IMS prepotential for the Lie algebra USp(4) with $n_{\bf 4}$ hypermultiplets in the fundamental representation, $n_{\bf 5}$ in the traceless antisymmetric  representation,   and $n_{\bf 10}$ in the adjoint representation is 
\begin{equation}	
\begin{aligned}
6\mathcal{F}_{\text{IMS}}= &-(8 n_{\bf 10} + n_{\bf 4} - 8) \phi_2^3  - 
 8(n_{\bf 10} +  n_{\bf 5} - 1) \phi_1^3 \\
&-3(4 n_{\bf 10}+ n_{\bf 4} - 4 n_{\bf 5} - 4) \phi_1^2 \phi_2  + 
 3(6 n_{\bf 10}+ n_{\bf 4} - 2 n_{\bf 5} - 6)  \phi_1 \phi_2^2. 
\end{aligned}
\end{equation}
\end{thm}
			
	The simple roots are $[2,-1]$ and $[-2,2]$ in the basis of fundamental weights. 
	We take $(\phi_1,\phi_2)$  in the basis of coroots. The dual open Weyl chamber is defined by 
	\begin{equation}
 \varphi_2 - \varphi_1>0, \quad   2\varphi_1-\varphi_2>0.
	\end{equation}				
	The two polynomial $\mathcal{F}_{\text{IMS}}$ and $\mathcal{F}$  match for $\alpha_0 =0, \alpha_1\leftrightarrow \phi_2$, $\alpha_2\leftrightarrow \phi_1$ and : 
	\begin{equation}
	n_{\bf 5}= S^2+1 - n_{\mathbf{10}}, \quad n_{\bf 4} = 4(S^2+4 - 2 n_{\mathbf{ 10}}- 2 g).
	\end{equation}	
	We see in particular that $n_{\mathbf{ 10}}$ is not fixed by the matching of the cubic polynomials but is only restricted by inequalities: 
		\begin{equation}
0\leq  n_{\mathbf{10}} \leq 2 - g +\frac{1}{2}S^2, \quad 0\leq n_{\mathbf{10}}\leq  1+S^2,	\end{equation}
The assumption $n_{\bf{10}}\geq 0$ and $n_{\bf{5}}\geq 0$ gives a  lower bound on  the self-intersection $S^2$:
\begin{equation}
-1 \leq S^2.
\end{equation}
Following Witten's quantization argument, the number of hypermultiplets transforming in the adjoint representation is given by the genus $g$ of the curve $S$, and we have 
\begin{equation}
n_{\bf{10}}=g, \quad n_{\bf{5}}=S^2+1-g, \quad n_{\bf{4}}= 4S^2 + 16(1-g). 
\end{equation}

  \section{F-theory  on a Calabi-Yau threefold  
and anomaly  cancellation}
In this section, we prove that the Calabi--Yau threefold USp($4$)-model engineers anomaly free 6D $\mathcal{N}=(1,0)$ supergravity theories. 

\subsection{Generalized Green-Schwarz mechanism in F-theory}

We assume that $Y$ is  a simply-connected elliptically fibered Calabi-Yau threefold with holonomy SU($3$). The SU($3$)  holonomy and the simply-connected  condition forces  $B$ to be a rational surface.  
The low energy effective description of F-theory compactified on an elliptically fibered Calabi-Yau threefold $Y$ with a base $B$ is six-dimensional $\mathcal{N}=(1,0)$ supergravity coupled to 
$n_{\text T}=h^{1,1}(B)-1$ tensor multiplets, $n_{\text H}^0=h^{2,1}(Y)+1$ neutral hypermultiplets, and  $n_{\text V}^{(6)}$ massless vector multiplets such that  $n_{\text V}^{(6)}+n_{\text T}=h^{1,1}(Y)-2$ \cite{Morrison:1996pp}.

Consider a 6D ${\cal N}=(1,0)$ supergravity theory with $n_{\text T}$ tensor multiplets, $n_{\text V}^{(6)}$ vector multiplets, and $n_{\text H}$ hypermultiplets. 
We assume that the gauge group is simple, so that 
\begin{equation}
n^{(6)}_V= \dim G. 
\end{equation}

We distinguish between the numbers $n_{\text H}^0$ of neutral hypermultiplets and n$_{\mathbf{R_i}}$ of both charged and neutral hypermultiplets transforming in the representation $\mathbf{R}_i$ of the gauge group. 
CPT invariance requires that $\textbf{R}_i$ is a quaternionic representation \cite{Schwarz:1995zw}. When $\bf{R}_i$ is pseudo-real, we can have half-hypermultiplets transforming under $\bf{R}_i$, which can give half-integer values for $n_{\bf{R}_i}$. 
We also count as neutral any hypermultiplet whose  charge is given by the zero weight of a representation. 
We denote by $\dim{\mathbf R}_{i,0}$ the number of zero weights in the representation ${\mathbf R}_{i}$. 
The total number of charged hypermultiplets is then \cite{GM1}
\begin{equation}
n_{\text H}^{\text{ch}}=\sum_i (\dim \mathbf{R}_i -\dim{\mathbf R}_{i,0}) n_{\mathbf{R_i}}.
\end{equation}
The total number of hypermultiplets is  $n_{\text H}=n_{\text H}^0+n_{\text H}^{\text{ch}}$.  
The pure gravitational anomaly  is cancelled by the vanishing of the coefficient of $\mathrm{tr}\  R^4$ in the anomaly polynomial \cite[Footnote 3]{Salam}:
\begin{equation}
\label{eqn:grav}
n_{\text H}-n_{\text V}^{(6)}+29 n_{\text T}-273=0.
\end{equation}
 Using the duality between F-theory on an elliptically fibered Calabi-Yau threefold with base $B$ and type IIB on $B$, Noether's formula implies the following for the number of tensor multiplets \cite{Sadov:1996zm}:
\begin{equation}
n_{\text T}=h^{1,1}(B)-1=9 - K^2.
\end{equation}
The remaining part of the anomaly polynomial is \cite{Schwarz:1995zw}:
 \begin{equation}
\mathcal  I_8= \frac{9-n_{\text T}}{8} (\mathrm{tr} \   R^2)^2+\frac{1}{6}\  X^{(2)} \mathrm{tr}\ R^2-\frac{2}{3} X^{(4)},
 \end{equation}
where 
\begin{align}
X^{(n)}=\mathrm{tr}_{\mathbf{adj}}\  F^n -\sum_{i}n_{\mathbf R_i} \mathrm{tr}_{\mathbf R_i}\  F^n.
\end{align}
 For a  reference representation $\mathbf F$,  the trace identities for a representation $\mathbf{R}_{i}$ of a simple group $G$ are
\begin{equation}
\tr_{\bf{R}_{i}} F^2=A_{\bf{R}_{i}} \tr_{\bf{F}} F^2 , \quad \tr_{\bf{R}_{i}} F^4=B_{\bf{R}_{i}} \tr_{\bf{F}} F^4+C_{\bf{R}_{i}} (\tr_{\bf{F}} F^2)^2
\end{equation}
with respect to a reference representation $\bf{F}$ for each simple component $G$ of the gauge group.\footnote{We denote this representation by $\bf{F}$ as we have chosen the fundamental representation(s) for convenience. However, any representation can be used as a reference representation.} The coefficients $A_{\bf{R}_{i}}$, $B_{\bf{R}_{i}}$, and $C_{\bf{R}_{i}}$ depends on the gauge groups and are listed in \cite{Erler,Avramis:2005hc,vanRitbergen:1998pn}. We have 
\begin{align}
X^{(2)}&=\Big(A_{\textbf{adj}}   -\sum_{i}n_{\mathbf R_i} A_{\mathbf R_i}\Big) \mathrm{tr}_{\mathbf F}  F^2\\
X^{(4)}&=\Big(B_{\textbf{adj}}   -\sum_{i}n_{\mathbf R_i} B_{\mathbf R_i}\Big) \mathrm{tr}_{\mathbf F}  F^4 +
\Big(C_{\textbf{adj}}   -\sum_{i}n_{\mathbf R_i} C_{\mathbf R_i}\Big)( \mathrm{tr}_{\mathbf F} F^2)^2
.
\end{align}
  If  $G$  does not have two independent quartic Casimir invariants, we take $B_{\bf{R}_{i}}=0$ \cite{Sadov:1996zm}.
If the simple group $G$ is supported on a divisor $S$ and $K$ is the canonical class of the base of the elliptic fibration,
we can  factor $\mathcal{I}_8$ as a perfect square following Sadov's analysis \cite{Sadov:1996zm,Sagnotti:1992qw}: 
 \begin{align}
\mathcal I_8 &= \frac{1}{2}(\frac{1}{2} K\mathrm{tr} \   R^2-\frac{2}{\lambda} S\ \mathrm{tr}_{\mathbf F} F)^2.
 \end{align}
Sadov showed this factorization matches the general expression of $\mathcal  I_8$  if and only the following anomaly cancellation conditions hold  \cite{Sadov:1996zm} (see also \cite{GM1,Park}):
\begin{subequations}\label{eq:AnomalyEqn}
\begin{align}
n_{\text H}-n_{\text V}^{(6)}+29n_{\text T}-273 &=0,\label{Grav.an}\\
n_{\text T}&=9-K^2 , \\
\left(B_{\bf{adj}}-\sum_{i}n_{\bf{R}_{i}}B_{\bf{R}_{i}}\right)& = 0, \label{B.an}\\
\lambda \left(A_{\bf{adj}}-\sum_{i}n_{\bf{R}_{i}}A_{\bf{R}_{i}}\right) & =6  K S, \\
\lambda^2 \left(C_{\bf{adj}}-\sum_{i}n_{\bf{R}_{i}}C_{\bf{R}_{i}}\right) & =-3 S^2.
\end{align}
\end{subequations}
The coefficient  $\lambda$ is a  normalization factor  chosen such that the  smallest topological charge of an embedded SU($2$) instanton in $G$ is one \cite{Kumar:2010ru, Park, Bernard}. This forces $\lambda$ to be the Dynkin index of the fundamental representation of  $G$ as summarized in Table \ref{tb:normalization} \cite{Park}. 

Using adjunction ($KS+S^2=2g-2$), the last two anomaly equations give an expression for the genus of $S$:
\begin{equation}\label{Witten.an}
\lambda \left(A_{\bf{adj}}-\sum_{i}n_{\bf{R}_{i}}A_{\bf{R}_{i}}\right) -2\lambda^2 \left(C_{\bf{adj}}-\sum_{i}n_{\bf{R}_{i}}C_{\bf{R}_{i}}\right) =12  (g-1).
\end{equation}

\begin{table}[htb]
\begin{center}
\begin{tabular}{|c|c|c|c|c|c|c|c|c|c|}
\hline
 $\mathfrak{g}$ & A$_n$ ($n\geq 1$) & B$_n$  ($n\geq 3$) & C$_n$ ($n\geq 2$) & D$_n$  ($n\geq 4$)& E$_8$ & E$_7$ & E$_6$&  F$_4$ & G$_2$ \\
 \hline
 $\mathbf{F}$ & $\mathbf{n+1}$ & $\mathbf{V}$=$\mathbf{2n+1}$ & $\mathbf{V}$=$\mathbf{2n}$  & $\mathbf{V}$=$\mathbf{2n}$ & $\mathbf{248}$ & $\mathbf{56}$ & $\mathbf{27}$ & $\mathbf{26}$ & $\mathbf{7}$ \\
 \hline 
 $\lambda=I_2(\mathbf{F})$ & $1$ & $2$  & $1$ & $2$ & $60$ & $12$ & $6$ & $6$ & $2$ \\
 \hline  
\end{tabular}
\caption{ The normalization factors for each simple Lie algebra used in the 6D anomaly equations \cite{Kumar:2010ru, Park}. 
The second row identifies the defining representation $\mathbf{F}$ for each algebra; the representation $\mathbf{F}$ is the smallest non-trivial representation of the Lie algebra. 
The coefficient  $\lambda$ is  the second Dynkin index of $\mathbf{F}$. The second Dynkin of fundamental representations of a simple Lie algebra were first computed in  \cite[Table 5]{Dynkin.SubA}. 
\label{tb:normalization}}
\end{center}
\end{table}

\subsection{Anomaly cancellation for USp($2n$) with matter in the adjoint, fundamental, and traceless antisymmetric representation}
Wel now consider the specificity of the USp($2n$)-model.  We will need the following trace identities: 
\begin{table}[hbt]
\begin{center}
$
\begin{array}{|c|c|c|c||c|c|c|}
\hline 
\bf{R} & A_{\bf{R}} & B_{\bf{R}} & C_{\bf{R}} & \dim \bf{R} & \dim \bf{R}_0& \dim \bf{R} -\dim \bf{R}_0 \\
\hline
\bf{F} &       1 &    1  &      0& 2n  & 0 & 2n\\  
\hline
\bf{Adj} &   2n+2  &  2n+8   &  3 & n(2n+1)& n& 2n^2 \\
\hline
{\bf{\Lambda}}^2_0&      2n-2 &   2n-8  &   3  & (n-1)(2n+1) & n-1 & 2 n (n-1)  \\
\hline 
\end{array}
$
\caption{Coefficients for the trace identities in the case of  USp($2n$). The reference representation $\mathbf{F}$ is the  representation $\bf{2n}$. 
 \label{table:SONormalization}  See \cite{PVN,Schwarz:1995zw,Sagnotti:1992qw}. }
\end{center}
\end{table}

We first ignore the condition for the cancellation of the gravitational anomaly, namely equation \eqref{Grav.an}. 
After fixing the conventions for the trace identities and the coefficient $\lambda=1$,  we are left with three linear  equations that have a unique solution:
\begin{equation}
n_{\bf{adj}}=g,\quad n_{\bf{F}}=16(1-g)+2(4-n) S^2,\quad   n_{{\bf{\Lambda}}^2_0}=1-g+S^2.
\end{equation}
We are left with the pure gravitational  anomaly, which requires checking  equation  \eqref{Grav.an}.  Since we have explicit expressions for the number of charged hypermultiplets, this is a straightforward computation.
We then see immediately that  equation \eqref{Grav.an} is also satisfied and  $  \mathcal{I}_8 $ factors as a perfect square:
\begin{align}
  \mathcal{I}_8 &= \frac{9-n_T}{8} (\mathrm{tr} \   R^2)^2+\frac{1}{6}X^{(2)}\ \mathrm{tr}\ R^2\ -\frac{2}{3} X^{(4)}\\
 &=\frac{1}{2} ( \frac{1}{2} K \mathrm{tr}\  R^2+ 2S\ \mathrm{tr}_{\mathbf{F}}\, F^2)^2
   \end{align}
	
	\subsection{Geometry of representation  multiplicities}
	In this section, we study the geometric meaning of the number $n_{\mathbf{10}}$, $n_{\mathbf{5}}$, and $n_{\mathbf{4}}$ computed from triple intersection numbers and anomaly cancellation conditions. 
	The matching of the triple intersection numbers with the prepotential does not restrict $n_{\mathbf{10}}$.
	We get $n_{\bf{10}}=g$ by using  Witten's quantization argument in 5D or by the anomaly cancellation conditions in the 6D theory. 
	
In this subsection, we would like to give a geometric description of $n_{\bf{5}}$ and $n_{\bf{4}}$. 
Using Riemann-Roch, we can express $n_{\mathbf{5}}$ as a difference of homomorphic Euler characteristics.

We denote by $\chi(M,\mathscr{F})$ the holomorphic Euler characteristic of the sheaf $\mathscr{F}$ in the variety $M$. 
 By definition 
 \begin{equation}
 \chi(M,\mathscr{F})=\sum_{i=0}^{\mathrm{dim} M} \mathrm{dim}\ H^i(M, \mathscr{F}).
 \end{equation}
 \begin{thm}
 For  a USp($2n$)-model with a gauge divisor $S$ of genus $g$ in a smooth compact rational surface $B$, we have 
 \begin{equation}
 n_{{\bf{\Lambda}}^2_0}=\chi(B, \mathscr{O}_B(S))-1.
\end{equation}
 \end{thm}
 \begin{proof}
The Riemann-Roch theorem states that for a curve $S$ in an algebraic surface $B$: 
\begin{equation}
\chi(B,\mathscr{O}_S)=\frac{1}{2}S(S-K)+\chi(\mathscr{O}_B),
\end{equation}
which can be rewritten as 
\begin{equation}
\chi(B,\mathscr{O}_S)=S^2-\frac{1}{2}S(S+K)+\chi(B,\mathscr{O}_B),
\end{equation}
Since $-S(S+K)=2-2g$, we have 
\begin{equation}
\chi(B,\mathscr{O}_S)=S^2+1-g +\chi(B,\mathscr{O}_B).
\end{equation}
Hence 
\begin{equation}
 n_{{\bf{\Lambda}}^2_0}=S^2+1-g=\chi(B, \mathscr{O}_B(S))-\chi(B,\mathscr{O}_B)=\chi(B, \mathscr{O}_B(S))-1.
\end{equation}	
In the last equality, we use that  $B$ is a rational surface and therefore has  holomorphic Euler characteristic $1$.	
\end{proof}
	\begin{rem}
		If $S'$ is the double cover of $S$ branched at the points $V(S,a_2)$, then  
		 as can be seen by computing the topological Euler characteristic of $S$, we have\footnote{To compute the genus of $S'$, one can first compute its topological Euler characteristic and use the formula $2-2g(S')=\chi(S')$. 
Since $S'$ is a double cover, its   topological Euler characteristic of $S'$ is $\chi(S')=2 \chi(S)-b$ where $b=-2 K \cdot S$ is the number of branched points. 
Using $\chi(S)=-K\cdot S-S^2$, it follows that 
$\chi(S')=-2S^2$ and 
 $g(S')=S^2+1$.}: 
	\begin{equation}
	g(S')=1+S^2.
	\end{equation}
	Hence we can also write $n_{\bf 5}$ as the difference:
	\begin{equation}
 n_{{\bf{\Lambda}}^2_0}=S^2+1-g=g(S')-g(S).
	\end{equation}
The curve $S'$ is defined by 
The divisor $D_1$ which can be described as a conic bundle over $S'$ with only double lines as singular fiber and a discriminant locus $V(a_2)$. 
 		
\end{rem}

	The number $ n_{\mathbf{V}} $  can be written as the intersection number between the curve $S$ and the locus $V(b_{8,2n})$ of class $-(8K+2nS)$:
	\begin{equation}
	 n_{\mathbf{V}} =-(8K+2nS)\cdot S.
	 \end{equation}
	This is in line with the traditional wisdom from classical intersecting  brane models. 
		However,  if $v(a_4)>2$ (resp. $v(a_6)>4$), the locus over which the representation $\mathbf {V}$ is localized is $V(a_{6,2n})\cap S$  (resp. $V(a_{4,n})\cap S$) and the previous interpretation of $n_{\bf V}$ losses its geometric meaning. 
	But a new geometric matching emerges in each case. 	
	\begin{enumerate}
\item	When $v(a_4)>n$, the fiber over $V(a_2)$ has the node $C_2^{\pm}$ both of which has has  weight $[1,-1]$, which is a weight of the  representation $\mathbf{V}$. 
	That means that we have to add the intersection numbers $-2 S\cdot K$ to $(-6K-2nS)\cdot S$ and we retrieve again the interpretation in terms of  intersecting brane literature:
	\begin{equation}
	 n_{\mathbf{V}} =(-2K)\cdot S + (-6K-2nS)\cdot S= (-8K-2nS)\cdot S.
	\end{equation}
	\item
	In the case of $v(a_6)>2n$, the curve $S$ intersects the other components of the discriminant locus over points of multiplicities $2$. If each of these points is counted twice to take into account the scheme structure, we again retrieve the geometric meaning of $n_{\mathbf V}$: 
	\begin{equation}
	 n_{\mathbf{V}} =2\Big( (-4K-nS)\cdot S\Big)= -(8K+2nS)\cdot S.
	\end{equation}
\end{enumerate}

\subsection{An heterotic sanity check}

In the case of a USp($4$)-model, we will check the counting on a model with a heterotic dual. 
	If the curve $S$ has genus $g=0$, then:
	\begin{equation}
	g=0\Longrightarrow n_{\mathbf{10}}=0, \quad n_{\bf 5}=S^2+1, \quad \text{ and} \quad  n_{\bf 4} = 4(S^2+4).
	\end{equation}
	 These values  match the matter content of a  Spin($5$) gauge theory in ${\cal N}=1$ six dimensional supergravity  obtained  by compactifying the E$_8\times$E$_8$  Heterotic string theory on a K3 surface with  instanton numbers $(12+n,12- n)$ where $n=n_{\bf{5}}-1=S^2$.

\subsection{Frozen representations}
In this section, we discuss geometric configurations which, despite containing curves corresponding to weights in a particular representation, are nonetheless forced by anomaly cancellation to have zero hypermultiplets charged under that representation.

For non-simply-laced Lie algebras, the numbers of representations are not given by the numbers of collision points supporting the singular fibers carrying the weights of these representations. It follows that it is possible to have a vanishing number of hypermultiplets carrying a representation despite the fact that the geometry has curves carrying the weights of that representation. In such a case, we say that the representation is frozen.

As an example, we consider a curve $S$ reaching the lower bound for $S^2$, that is, $S$ is a rational curve of self-intersection $-1$, so that
\begin{equation}
(g=0 \quad  \text{and}\quad S^2=-1)\Longrightarrow g=n_{\mathbf{10}}=n_{\mathbf{ 5}}=0,\quad  \text{and} \quad n_{\mathbf{ 4}}=12.
\end{equation} 
It follows that when $S$ is an exceptional rational curve, there is no matter in the adjoint nor in the traceless antisymmetric representation. 
In that case, the fundamental representation has multiplicity $12$. 
Still, on the curve $S$, there are six points over which the I$_4^{\text{ns}}$ fiber degenerates and the degenerated fibers contain vertical rational curves  whose weights are in the representation $\mathbf{5}$. 
However, the anomaly cancellation gives  $n_{\mathbf{ 5}}=0$. We say that the representation  $\mathbf{5}$ is {\em frozen}.

More generally, the representation $\mathbf{5}$ of a USp(4) model is frozen if and only if  the following equivalent conditions holds:
\begin{align}
S^2 &=g-1, \quad 
\chi(B, \mathscr{O}_B(S)) =1, \quad 
\chi(S) =-2K\cdot S, \quad
K\cdot S =g-1, \quad
K\cdot S =S^2.
\end{align}
\begin{rem}
If $S$ is a smooth rational curve, the representation $\mathbf{5}$ is frozen if and only if $S^2=-1$. 
That means that the representations $\bf{10}$ and $\bf{5}$ are both  frozen if and only if $S$ is an exceptional divisor of the base and can be contracted to  a point without introducing a singularity in the base. 
\end{rem}

\begin{rem}
For USp(4), the only representations that can be frozen are the adjoint representation (when $g=0$) and the traceless antisymmetric representation $\Lambda^2_0$ (the $\mathbf{5}$)  when  $S^2=g-1$. 
The representation $\mathbf{4}$ cannot be frozen since $n_{\mathbf{4}}$ vanishes exactly when the locus over which the curves with  weights in the representation $\mathbf{4}$ are absent. 

\end{rem}

The USp($4$)-model illustrates several surprising phenomena in F-theory such as the existence of frozen representations and the role of the dimension of the base when determining representations.

\section{ Katz--Vafa method and  USp($4$)-models}
  The Katz--Vafa method is a heuristic technique to determine the matter representations appearing in presence of an enhancement of a Lie algebra $\mathfrak{h}$ into a Lie algebra $\mathfrak{g}$  using branching rules expressing the decomposition of the  adjoint representation of $\mathfrak{g}$ along irreducible representations of $\mathfrak{h}$.  
  In F-theory, $\mathfrak{h}$ is the Lie algebra determined by the fibers over codimension-one points while $\mathfrak{g}$ is associated with the Dynkin diagram of a singularity enhancement over codimension-two points. 
  A more systematic and reliable method is to compute the weights of the rational curves  appearing in the fiber enhancement to determine the representations they correspond to by using the notion of {\em saturation of weights}.  
Nevertheless, it is instructive to know how and why the Katz--Vafa method fails.

\subsection{Branching rules}

{\em Branching rules} describe how an irreducible representation of a Lie algebra $\mathfrak{g}$ decomposes into irreducible representations of one of a subalgebra $\mathfrak{h}$. 
A subalgebra $\mathfrak{h}$ can be embedded in a Lie algebra $\mathfrak{g}$ in more than one way, which might  lead to inequivalent branching rules.  
Thus, identifying a subalgebras (up to isomorphisms)  $\mathfrak{h}$ of $\mathfrak{g}$ is not enough to determine a branching rule, even if $\mathfrak{h}$ is a maximal subalgebra of $\mathfrak{g}$. 
What is required to determine a branching rule is an embedding of $\mathfrak{h}$ in $\mathfrak{g}$. For example,  Dynkin's classification of maximal subalgebras shows that the Lie algebra of type E$_8$  has three  types of  maximal  A$_1$ subalgebras leading to distinct branching rules.

Branching rules were first introduced in the first part of the 20th century to answer questions in quantum mechanics such as the study of the degeneration levels of the hydrogen atom. Branching rules quickly imposed themselves as fundamental tools in nuclear physics,  particle physics, and condensed matter physics. With the development of gauge theories,  the paradigm of symmetry breaking,  Grand Unified Theories, and the  Kaluza--Klein dimensional reduction via compactification,  branching rules became an irreplaceable instrument in the  toolkit of theoretical and mathematical physicists.

    The first examples of branching rules can be traced back to Hermann Weyl's seminal book on quantum mechanics and representation theory.
     Weyl gave a beautiful complete branching theorem for the restriction of an irreducible representation of $A_{n}$  to a sum of irreducible representations of $A_{n-1}$. Murnaghan described the branching rules for the restriction of the special orthogonal group SO($2n+1$) to its subgroup SO($2n$) and USp($2n+2$) to USp($2n$).
Littlewood and Murnaghan gave the branching rules for the restriction of U($2n$) to USp($2n$). 
  Many applications of branching rules are reviewed by Hamermesh.
Branching rules are collected in tables such as in the book of McKay and Patera \cite{McKayPatera},  the review of Slansky on Grand Unification \cite{Slansky:1981yr} (see also \cite{Yamatsu:2015npn}),  and many can be also explored in computer packages such as Sage and LieArt.  
  \subsection{Branching rules in F-theory}
  
 We denote by $\mathfrak{g}\downarrow \mathfrak{h}$  the restriction of  a Lie algebra $\mathfrak{g}$ to a subalgebra $\mathfrak{h}$. 
 If a gauge theory with gauge algebra $\mathfrak{g}$ and matter (a Higgs field) transforming in the adjoint representation of $\mathfrak{g}$ undergoes a symmetry breaking 
  to a subalgebra $\mathfrak{h}$, Golstone bosons will transform  under irreducible representations $\mathbf{R}_i$ appearing in the branching rule for the  adjoint representation of  $\mathfrak{g}$ under the restriction of $\mathfrak{g}\downarrow \mathfrak{h}$:
  \begin{equation}\label{eq:KV}\mathfrak{g}\downarrow \mathfrak{h}:\quad\quad 
\bf{adj}  (\mathfrak{g})=\bf{adj}  (\mathfrak{h})\bigoplus_i \bf{R}_i.
  \end{equation}
  In F-theory, codimension-one singular fibers determine a gauge algebra  $\mathfrak{h}$ and the codimension-two singularities would give a gauge algebra $\mathfrak{g}$ if they were lifted to divisors by a blowup. In this situation, we say that $\mathfrak{h}$ is  enhanced to $\mathfrak{g}$. Typically, $\mathfrak{h}$ is a subalgebra of $\mathfrak{g}$.

 In this situation, the {\em Katz--Vafa method} proposes that the matter representation consists of the sum of  the irreducible representations  $\mathbf{R}_i$ appearing 
 in equation \eqref{eq:KV} for  the branching rules for the  adjoint representation when the Lie algebra  $\mathfrak{g}$ is reduced to the subalgebra $\mathfrak{h}$. 
  The Katz--Vafa method was   originally considered for the following ten types of rank-one enhancements between ADE Lie algebras: 
    \begin{align}
& A_n \to A_{n-k} \oplus A_{k-1},\\
& D_n \to  D_{n-1}, \   A_{n-1},\  
D_{n-k}\oplus A_{k-1},  \\
& E_6\to D_5,\ A_5,\\
& E_7 \to E_6, \ 
D_6, \    A_6, \\	
& E_8 \to E_7. 		
\end{align}
A detailed analysis of these branching rules are presented in \cite{EK.KV}.

 \subsection{USp($4$)-models as kryptonite for the Katz--Vafa method}\label{Sec:Kryptonite}

 The  representations of Spin($5$) of lowest dimensions are the trivial representation $\bf{1}$, 
the spin representation $\bf{4}$, 
the vector representation $\bf{5}$, the adjoint representation $\bf{10}$, and the traceless symmetric tensor representation $\bf{14}$.   
  
    In this section, we assume that the base of the USp($4$)-model is a surface. A direct consequence of that assumption is that all the fibers over codimension-two points are geometric and therefore split.
A generic USp($4$)-model has two codimension-two loci where the I$_4^{\text{ns}}$ fiber degenerates (see Figure \ref{Fig.V1}). 
 Over $S\cap V(4a_2 a_{6,4} -a_{4,2}^2)$, the fiber  I$_4^{\text{ns}}$   degenerates to a fiber of type I$_5$ containing a rational curve carrying the weights of the representation $\bf{4}$. Over 
 $S\cap V(a_2)$, the  fiber  I$_4^{\text{ns}}$  degenerates to an incomplete fiber of type I$_0^*$ carrying the weights of the representation $\mathbf{5}$.
 The fibers I$_5$ and I$_0^*$ correspond respectively to the affine Dynkin diagram for A$_4$ and D$_4$. 
 
 We would like to know if the branching rules  $A_4\downarrow C_2$ and $D_4\downarrow C_2$  for the adjoint representation will output the  representations $\mathbf{5}$ and $\mathbf{4}$ that we see geometrically by computing the weights of vertical rational curves appearing over the codimension-two loci  $S\cap V(4a_2 a_{6,4} -a_{4,2}^2)$ and   $S\cap V(a_2)$:
   \begin{align}
   &  (A_4\downarrow C_2) \overset{?}{\implies} \mathbf{4},\\
&  (D_4\downarrow  C_2)\overset{?}{\implies}\mathbf{5}.
  \end{align}
If that is the case, we will say that the Katz--Vafa method predicts the matter content $\mathbf{4}$ or $\mathbf{5}$ of the USp($4$)-model.

The Lie algebra A$_4$ has three types of  maximal subaglebras, namely 
A$_3$ (obtained by removing the node $\alpha_1$ or $\alpha_4$), A$_1$$\times$A$_2$ (obtained by removing one of the two interior  nodes $\alpha_2$ or $\alpha_4$), 
and C$_2$, which is an $S$-algebra. 
\begin{equation}
A_4:  \quad   A_3, \quad A_1\times A_2, \quad C_2.
\end{equation}

The algebra C$_2$ can be embedded in different ways as a subalgebra of A$_4$ since C$_2$ is also maximal subalagebra of A$_3$ 

Under the reduction $A_4\downarrow C_2$  for a  maximal subalgebra C$_2$, we have the following branching rule for the adjoint representation:
\begin{equation}\label{Eqn:14}
A_4\downarrow C_2:\quad 
\bf{24}=\bf{10}\oplus \bf{14}.
\end{equation}
The representation $\bf{14}$ is irreducible and corresponds to $\Gamma_{0,2}$. 
It is clear that this embedding does not produce the representation $\bf{5}$ that we see geometrically. 

The Lie algebra of type  A$_4$ also contains C$_2$ as a non-maximal subalgebra appearing in the  sequence of reduction combination  A$_4\downarrow$ A$_3\downarrow$ C$_2$. In that case, the branching rule for the adjoint representation can be derived by  steps and gives\footnote{
We can understand this branching rule as follows \cite{EK.KV}. 
Under the reduction $A_4\downarrow A_3$,  the adjoint of A$_4$ decomposes as follows 
$$
A_4\downarrow A_3: \quad \bf{24}= \bf{15}\oplus \bf{4}\oplus\overline{\bf{4}}\oplus\bf{1}.
$$
Under the reduction $A_3\downarrow C_2$, the adjoint, the fundamental, and the anti-fundamental of A$_3$ decomposes as 
$$
A_3\downarrow C_2:\quad 
\bf{15}\to \bf{10}\oplus \bf{5}, \quad \bf{4}\to \bf{4}, \quad  \bf{\overline{4}}\to \bf{4}.
$$
Hence, under A$_4\downarrow$ A$_3\downarrow$ C$_2$, we have 
$$
A_4\downarrow A_3\downarrow C_2:\quad
\bf{24}= \bf{10}\oplus  \bf{5}\oplus\bf{4}\oplus{\bf{4}}\oplus\bf{1}.
$$
} 
\begin{equation}
A_4\downarrow A_3\downarrow C_2:\quad
\bf{24}= \bf{10}\oplus  \bf{5}\oplus\bf{4}\oplus{\bf{4}}\oplus\bf{1}.
\end{equation}
Here, we get both the representations $\bf{4}$ and $\bf{5}$. However, that is in contradiction with the geometry of the USp($4$)-model as at no point of the base we see both the weights of the representation $\bf{5}$ and the representation $\bf{4}$.

The Lie algebra C$_2$ can be embedded in D$_4$ along many inequivalent channels defined by sequences of  maximal embeddings: 
\begin{equation}
\begin{cases}
D_4\downarrow B_3 \downarrow( \mathfrak{u}_1\oplus C_2) \downarrow C_2, \quad 
D_4 \downarrow (A_1\oplus C_2)\downarrow C_2, \\ 
D_4\downarrow B_3\downarrow A_3\downarrow C_2, \quad   \quad \qquad
D_4\downarrow (A_3\oplus \mathfrak{u}_1)\downarrow A_3\downarrow C_2
\end{cases}
\end{equation} 
We denote the adjoint of D$_4$ as $\mathbf{24}$. 
In each channel, we end up with the branching rule:
\begin{equation} 
\bf{28}= \bf{10}\oplus \bf{5}\oplus\bf{5}\oplus \bf{5}\oplus \bf{1}\oplus\bf{1}\oplus \bf{1}.
\end{equation}
This branching rule produces the representation $\bf{5}$ in accordance with the geometric analysis. 
It follows that in this case, the geometry is compatible with the Katz--Vafa method.

\subsection{SO($5$)-model }
The group USp($4$) is isomorphic to Spin($5$), the double-cover of the orthogonal group SO($5$). 
There are representations of USp($4$) that are only projective representations of SO($5$), that is particularly the case of the spin representation $\mathbf{4}$. 
The SO($5$)-model is studied in \cite{SO} and has the following matter representation: 
\begin{equation}
\mathbf{R}=\mathbf{10}\oplus\mathbf{5}.
\end{equation}
In contrast to USp($4$)-models, there is no representation $\mathbf{4}$ in the SO($5$)-model. 
It follows that while the generic USp($4$)-model has a matter content that is inconsistent with the Katz--Vafa method, the
 matter content of an SO($5$)-model is consistent with the Katz--Vafa method. 
 
  \subsection{General USp($2n$)-model }
 
 The matter content of a general USp($2n$)-model consists of the defining representation $\mathbf{V}$, the adjoint representation $\mathbf{Adj}=\text{Sym}^2_0 \textbf{V}$, and the traceless antisymmetric representation $\mathbf{\Lambda}^2_0$:
\begin{equation}
\mathbf{R}=  \mathbf{Adj}\oplus \mathbf{V}\oplus \mathbf{\Lambda}_0^2. 
\end{equation}
In a generic USp($2n$)-model we expect to see the weights of the representations $\mathbf{V}$ and $\mathbf{\Lambda}^2_0$  carried by rational curves components of singular fibers localized at different points in codimension-two. 
 
 The group USp($2n$) can also be embedded in SU($2n+1$) as a non-maximal subgroup sitting in a sequence of two 
 maximal embeddings for all $n\geq 2$. 
 For $n=2$, we also have in addition that C$_2$ can also be a  maximal subalgebra of A$_4$. And this gives the additional option that we discussed in equation \eqref{Eqn:14}.  
Thus, we have the following branching rules for $n\geq 2$ \cite{EK.KV}: 
\begin{align}
 A_{2n-1} \downarrow C_n &:& \quad 
 \mathbf{adj}\ \textup{SU(2n)}\   \    \    \     &= \mathbf{adj}\ \textup{USp(2n)}\oplus\mathbf{\Lambda}_0^2, \\
  A_{2n}\downarrow A_{2n-1} \downarrow C_n&:&\quad \mathbf{adj}\ \textup{SU(2n+1)} &= \mathbf{adj}\ \textup{USp(2n)}\oplus\mathbf{\Lambda}^2_0\oplus \mathbf{V}\oplus \mathbf{V}\oplus \mathbf{1}, \label{EqB2}\\
 A_{4}\downarrow C_2&:&\    \mathbf{adj}\ \textup{SU(5)} \quad\   \   &=\mathbf{adj}\ \textup{USp(4)} \oplus \mathbf{W}, 
\quad \mathbf{W}=\Gamma_{0,2}=\mathbf{14}\   .
\end{align}
 If $n=2$, we recall that $\mathbf{\Lambda}_0^2=\mathbf{5}$ and $\mathbf{V}=\mathbf{4}$,  and we are in the situation analyzed in section \ref{Sec:Kryptonite}.

The embedding  $ A_{2n-1}\downarrow C_n $ can describes the localized matter in the  representation  $\mathbf{\Lambda}_0^2$ and matches what we see in the geometry over $S\cap V(a_2)$. 
The sequence 
$A_{2n}\downarrow A_{2n-1} \downarrow C_n$ produces both $\mathbf{V}$ and $\mathbf{\Lambda}^2_0$ and therefore cannot correspond to localized matter as there are no points containing both weights.
The representation $\mathbf{14}$  coming from the branching rule for the maximal embedding $A_4\downarrow C_2$ does not match the geometry anywhere.

 Grassi and Morrison proposed using the  branching rule of equation \eqref{EqB2} to explain the matter content of  USp($2n$)-models \cite{GM1}. 
 The motivation behind the branching rules advocated in \cite{GM1} is  exploit the description of the Dynkin diagram of  $C_n$ of  USp($2n$) as the folding of the Dynkin diagram A$_{2n}$ of SU($2n+1$) via the $\mathbb{Z}/2\mathbb{Z}$ outer automorphism. 
  However, we point out that at no point over the base, we get both the representation $\mathbf{V}$ and $\mathbf{\Lambda}^2_0$.  
Thus, that would not explain the localized matter observed geometrically as illustrated in Figure \ref{fig:Sp4A}. 
The same conflict between the branching rules 
advocated in \cite{GM1} 
and fiber geometry appears for the G$_2$, F$_4$, and Spin($2n+1$)-models.  
Besides, when we are  the the branching point, 
We refer to \cite{EK.KV} for more information. 

 \clearpage

	\begin{figure}[htb]
	\begin{center}
			\scalebox{.8}{$
			\begin{array}{c}
			\begin{array}{c}
				\begin{tikzpicture}[scale=1]
					\draw[fill=black!8!] (0,0) ellipse (4.8cm and 2.2cm);
					\node at (2,-1.7) {\LARGE $B$};
\node at (-4.1,6) 
						{
							\begin{tikzpicture}[]
								\draw[dashed,thick,red] (-1.25,-1+.25) to (1.25,-1+.25);
								\draw[dashed,thick,red] (-1.25,-1-.25) to (1.25,-1-.25);
								\draw[dashed,thick,red]  (-1.25,-1+.25)  to (-1.25,-1-.25);
								\draw[dashed,thick,red]  (1.25,-1+.25)  to (1.25,-1-.25);
								\node[draw,ultra thick,circle,label=left:{$C_0$}](C0) at (0,0) {};
								\node[draw,ultra thick,circle,label=left:{$C_1^+$}](C1+) at (-1,-1) {};
								\node[draw,ultra thick,circle,label=right:{$C_1^-$}](C1-) at (1,-1) {};
								\node[draw,ultra thick,circle,label=left:{$C_2$}](C2) at (0,-2) {};;
								\draw[ultra thick] (C0) -- (C1+) -- (C2) -- (C1-) -- (C0);
							\end{tikzpicture}
						};

\draw[dotted,thick] (-4.2,-.2) --+ (0,5);
					\draw[dotted,thick] (-1.4,-.2) --+ (0,3.2);
									\draw[dotted,thick] (2.2,-.2) --++ (0,3);	
					\draw[dotted,thick] (4,1) -- ++(0,4.4);

					\node at (-1.4,4) 
						{
							\begin{tikzpicture}[scale=.8]
								\node[draw,ultra thick,circle,label=left:{$C_0$}](C0) at (0,0) {};
								\node[draw,ultra thick,circle,label=left:{$C_1^+\  $}](C1+) at (-1,-1) {};
								\node[draw,ultra thick,circle,label=right:{$ \ C_1^-$}](C1-) at (1,-1) {};
								\node[draw,ultra thick,circle,label=left:{$C_2^+\   $}](C21) at (-1,-2.5) {};
							  \node[draw,ultra thick,circle,label=right:{$ \   C_2^-$}](C22) at (1,-2.5) {};
								\draw[ultra thick] (C0) -- (C1+) -- (C21)--(C22) -- (C1-) -- (C0);
								
							\end{tikzpicture}
						};

					\node at (1.88,4) 
						{
							\begin{tikzpicture}
								\node[draw,ultra thick,circle,xshift=200,label=left:{$C_0$}](C0') at (0,0) {};
								\node[draw,ultra thick,circle,xshift=200,scale=.65,label=left:{$C_1^{'}$}](C13) at (0,-1) {2};
								\node[draw,ultra thick,circle,xshift=200,label=left:{$C_2$}](C2) at (0,-2) {};
								\draw[ultra thick] (C0') -- (C13) -- (C2);
							\end{tikzpicture}
						};
						
						\node at (-3.4,-.6) {$S$};
												\node at (-1.35,-1.4) {$\Delta'$};
												\filldraw (-1.35,-.35) circle (3pt);
												\filldraw (2.2,-.35) circle (3pt);
												\node at (3,-.7) {\scalebox{.8}{$S\cap V(a_2)$}};
												\node at (0,-.7) {\scalebox{.8}{$S\cap V(4a_2 a_{6,4}- a_{4,2}^2)$}};
						\draw[ultra thick] (-4.4,-.35) --++ (9,0);

						\node(a1) at (4,1) {};
					\node(b1) at (2.1,-.29) {};
					\node(c1) at (.1,1.4) {};
					\node(d1) at (-1.8,-1.4) {};
					 \draw[ultra thick] plot [smooth,tension=1] coordinates { (a1) (b1) (c1) (d1)};

					\node at (4,6) 
						{
							\begin{tikzpicture}[scale=2.4]
								\draw[scale=.5,domain=-1.2:1.2,variable=\x, ultra thick] plot({\x*\x-1,\x*\x*\x-\x-5});
							\end{tikzpicture}
						};
	
				\end{tikzpicture} 
			\end{array}\\\\
				\begin{array}{|l|c|}
				 \hline
				\text{Weierstrass model} &  zy^2=x^3 + a_2 x^2z + a_{4,2} s^2 x z^2+ a_{6,4} s^4z^3\\ \hline
				\text{Discriminant} & S\cap	\Delta' :\quad s= a_2^2(4a_2 a_{6,4}- a_{4,2}^2)=0 \\\hline
				\text{Matter representations} & \mathbf{10}\oplus\mathbf{5}\oplus\mathbf{4}  \\\hline
							\end{array}
			\end{array}
			$}
			\end{center}
	\caption{Geometry of the generic USp($4$)-model for a  two-dimensional  base $B$.	}
	\label{fig:Sp4A}
%
%
	\begin{center}
			\scalebox{1}{$
			\begin{array}{c}
			\begin{array}{c}
				\begin{tikzpicture}[scale=1]
					\draw[fill=black!8!] (0,0) ellipse (3.8cm and 1.7cm);
					\node at (-2.4,-.45) {$V(2s+a_2)$};
					\node at (2.4,-.45) {$V(2s-a_2)$};
					\node at (0,-1) {\LARGE $B$};
					\node at (3,.3) {$V(s)$};
					\draw[ultra thick] (-2.3,1) --++ (4.6,-2);
					\draw[ultra thick] (2.3,1) --++ (-4.6,-2);
					\draw[ultra thick] (-3,0) --++ (6,0);
					\draw[dotted,thick] (-1.3,.55) -- (-1.3,2.8);
					\draw[dotted,thick] (2.3,1) -- (2.3,3);
					\draw[dotted,thick] (-3,0) -- (-3,3.75);
					\draw[dotted,thick] (0,0) -- (0,3.75);
					\node at (-.3,5) 
						{
							\begin{tikzpicture}
								\node[draw,ultra thick,circle,xshift=200,label=left:{$C_0$}](C0') at (0,0) {};
								\node[draw,ultra thick,circle,xshift=200,scale=.65,label=left:{$C_1^{'}$}](C13) at (0,-1) {2};
								\node[draw,ultra thick,circle,xshift=200,label=left:{$C_2$}](C2) at (0,-2) {};
								\draw[ultra thick] (C0') -- (C13) -- (C2);
							\end{tikzpicture}
						};
					\node at (-3,5) 
						{
							\begin{tikzpicture}[]
								\draw[dashed,thick,red] (-1.25,-1+.25) to (1.25,-1+.25);
								\draw[dashed,thick,red] (-1.25,-1-.25) to (1.25,-1-.25);
								\draw[dashed,thick,red]  (-1.25,-1+.25)  to (-1.25,-1-.25);
								\draw[dashed,thick,red]  (1.25,-1+.25)  to (1.25,-1-.25);
								\node[draw,ultra thick,circle,label=left:{$C_0$}](C0) at (0,0) {};
								\node[draw,ultra thick,circle,label=left:{$C_1^+$}](C1+) at (-1,-1) {};
								\node[draw,ultra thick,circle,label=right:{$C_1^-$}](C1-) at (1,-1) {};
								\node[draw,ultra thick,circle,label=left:{$C_2$}](C2) at (0,-2) {};;
								\draw[ultra thick] (C0) -- (C1+) -- (C2) -- (C1-) -- (C0);
							\end{tikzpicture}
						};
					\node at (-1.3,3.1) 
						{
							\begin{tikzpicture}[scale=2]
								\draw[scale=.5,domain=-1.2:1.2,variable=\x, ultra thick] plot({\x*\x-1,\x*\x*\x-\x-5});
							\end{tikzpicture}
						};
					\node at (2.3,3.5) 
						{
							\begin{tikzpicture}[scale=2]
								\draw[scale=.5,domain=-1.2:1.2,variable=\x, ultra thick] plot({\x*\x-1,\x*\x*\x-\x-5});
							\end{tikzpicture}
						};
	
				\end{tikzpicture} 
			\end{array}\\\\
				\begin{array}{|l|c|}
				 \hline
				\text{Weierstrass model} & y^2 z = x^3 + a_2 x^2 z + s^2 x z^2\\ \hline
				\text{Discriminant} & 	\Delta =s^4 (4 s^2 - a_2^2) \\\hline
				\text{Matter representations} & \mathbf{10}\oplus\mathbf{5}  \\\hline
							\end{array}
			\end{array}
			$}
			\end{center}
	\caption{Summary of the geometry o the SO($5$)-model as studied in \cite{SO}.	}
	\label{fig:SO5}
	\end{figure}

\clearpage

\section*{Acknowledgements}
The authors are grateful to Ravi Jagadeesan,  Monica Jinwoo Kang, Sabrina Pasterski, and Julian Salazar for discussions. 
M.E. is supported in part by the National Science Foundation (NSF) grant DMS-1406925  and DMS-1701635 ``Elliptic Fibrations and String Theory''.  The research of P.J. 
is supported by the U.S. Department of Energy, Office of Science, Office of High Energy
Physics of U.S. Department of Energy under grant Contract Numbers DE-SC0012567
and DE-SC0019127.

\appendix 

\section{Weierstrass models and Deligne's formulaire }
\label{sec:Wmodel}
In this section, we  introduce the notation of  Deligne \cite{Deligne.Formulaire}. 
Let  $\mathscr{L}$ be a line bundle over  a normal quasi-projective variety  $B$.  We define the projective bundle of lines:
\begin{equation}
\pi: X_0=\mathbb{P}_B[\mathscr{O}_B\oplus \mathscr{L}^{\otimes 2}\oplus \mathscr{L}^{\otimes 3}]\longrightarrow B.
\end{equation} 
The relative projective coordinates of $X_0$ over $B$ are denoted $[z:x:y]$,  where $z$, $x$, and $y$ are defined  by the natural injection of 
 $\mathscr{O}_B$,   $\mathscr{L}^{\otimes 2}$, and $\mathscr{L}^{\otimes 3}$ into $\mathscr{O}_B\oplus \mathscr{L}^{\otimes 2}\oplus \mathscr{L}^{\otimes 3}$, respectively. Hence, 
  $z$ is a section of $\mathscr{O}_{X_0}(1)$, $x$ is a section of $\mathscr{O}_{X_0}(1)\otimes \pi^\ast \mathscr{L}^{\otimes 2}$, and
$y$ is a section of  $\mathscr{O}_{X_0}(1)\otimes \pi^\ast \mathscr{L}^{\otimes 3}$.

\begin{defn}
 A  Weierstrass model is an elliptic fibration $\varphi: Y\to B$  cut out by the zero locus of  a section of the  
line bundle $\mathscr{O}_{X_0}(3)\otimes \pi^\ast \mathscr{L}^{\otimes 6}$ in $X_0$. 
\end{defn}
The most general Weierstrass equation is written in the notation of Tate as
\begin{equation}
y^2z+ a_1 xy z + a_3  yz^2 -(x^3+ a_2 x^2 z + a_4 x z^2 + a_6 z^3) =0,
\end{equation} 
where $a_i$ is a section of $\pi^\ast \mathscr{L}^{\otimes i}$. 
The line bundle $\mathscr{L}$ is called the {\em fundamental line bundle} of the Weierstrass model $\varphi:Y\to B$. It can be defined directly from $Y$ as 
$\mathscr{L}=R^1 \varphi_\ast Y$. 
Following Tate and Deligne, we introduce the following quantities 
\begin{align}
\begin{cases}
b_2 &= a_1^2 + 4 a_2\\
b_4 &= a_1 a_3 + 2 a_4\\
b_6 &= a_3^2 + 4 a_6\\
b_8 &= a_1^2 a_6 - a_1 a_3 a_4 + 4 a_2 a_6 + a_2 a_3^2 - a_4^2\\
c_4 &= b_2^2 - 24 b_4\\
c_6 &= -b_2^3 + 36 b_2 b_4 - 216 b_6\\
\Delta &= -b_2^2 b_8 - 8 b_4^3 - 27 b_6^2 + 9 b_2 b_4 b_6\\
j& = {c_4^3}/{\Delta}
\end{cases}
\end{align}
The  $b_i$ ($i=2,3,4,6)$ and $c_i$   ($i=4,6$) are  sections of $\pi^\ast \mathscr{L}^{\otimes i}$. 
The discriminant $\Delta$ is a section of $\pi^\ast \mathscr{L}^{\otimes 12}$. 
They satisfy the two relations
\begin{align}
1728 \Delta=c_4^3-c_6^2, \quad 4b_8 = b_2 b_6 - b_4^2.
\end{align}
Completing the square in $y$ gives 
\begin{equation}
zy^2 =x^3 +\tfrac{1}{4}b_2 x^2 + \tfrac{1}{2} b_4 x + \tfrac{1}{4} b_6.
\end{equation}
Completing the cube in $x$ gives the short form of the Weierstrass equation
\begin{equation}
zy^2 =x^3 -\tfrac{1}{48} c_4 x z^2 -\tfrac{1}{864} c_6 z^3.
\end{equation}

\section{Representations of USp($2n$)}

The Lie algebra C$_\ell$ is studied in Planche III of \cite{Bourbaki.GLA46}.
The  Lie algebra C$_\ell$ has rank $\ell$, Coxeter number $2\ell$, dimension $\ell(2\ell +1)$, and  $\mathbb{Z}/2\mathbb{Z}$ center. 
Its root system consists of $2\ell^2$ roots, $2\ell$ of which are short roots. The Weyl group $W(\text{C}_\ell)$ of C$_\ell$ 
 is the semi-direct product of the symmetric group $\mathfrak{S}_\ell$ with the group  $(\mathbb{Z}/2\mathbb{Z})^\ell$, 
 hence, it has dimension  $2^\ell \ell!$

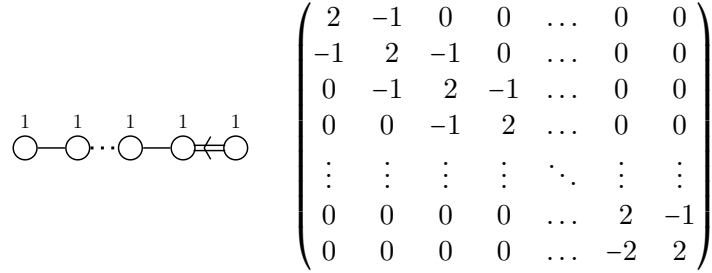
\begin{figure}[htb]
\begin{center}
\scalebox{.7}{$\begin{array}{c} \begin{tikzpicture}
				\node[draw,circle,thick,scale=1.25,label=above:{1}] (3) at (.8,0){};
				\node[draw,circle,thick,scale=1.25,label=above:{1}] (4) at (1.8,0){};
				\node[draw,circle,thick,scale=1.25,label=above:{1}] (5) at (2.8,0){};	
				\node[draw,circle,thick,scale=1.25,label=above:{1}] (6) at (3.8,0){};	
				\node[draw,circle,thick,scale=1.25,label=above:{1}] (7) at (4.8,0){};	
				\draw[thick]   (3) to (4);
				\draw[thick] (5) to (6);
				\draw[ultra thick, loosely dotted] (4) to (5) {};
				\draw[thick] (4.,-0.05) --++ (.6,0){};
				\draw[thick] (4.,+0.05) --++ (.6,0){};

				\draw[thick]
					(4.2,0) --++ (-60:.25)
					(4.2,0) --++ (60:.25);
			\end{tikzpicture}\end{array}$}
	\quad
	$
\begin{pmatrix}
\ 2 & -1 & 0 & 0 &\ldots   & 0 & 0\\
-1& \ 2 & -1 & 0 &\ldots  & 0 & 0\\
0 & -1 & \ 2 & -1 & \ldots & 0 & 0\\
0 & 0 & -1 & \ 2 & \ldots & 0 & 0\\
\vdots & \vdots & \vdots & \vdots &\ddots & \vdots & \vdots \\
0 & 0 & 0 & 0 & \ldots & \ 2 & -1\\
0 & 0 & 0 & 0 & \ldots & -2 & 2 
\end{pmatrix}
$
	\end{center}
	\caption{Dynkin diagram and Cartan matrix of  C$_\ell$.}
		\end{figure}

\begin{table}[htb]
\begin{center}
\begin{tabular}{|c|c|c|c|c|}
\hline
Reps  & Tableau & Highest weight & Dimension &Properties \\
\hline
& & & & \\
$\bf{V}$ &  \scalebox{.6}{\ydiagram{1}}  & $(1,0,0,\ldots, 0)$ & $2n$ & Pseudo-real, minuscule \\
& & & & \\
\hline
& & & & \\
${\mathbf{\Lambda}}_0^2 $ &  \scalebox{.6}{\ydiagram{2}}&$(0,1,0,\ldots, 0)$ &$(n-1)(2n+1)$& Real \\
& & &&  \\
\hline
& & & & \\
$\text{\bf{Sym}}^2_0 \bf{V}=\mathbf{adj}$ &  \scalebox{.6}{\ydiagram{1,1}} &$(2,0,0,\ldots,0)$ & $n(2n+1)$& Real \\
& & & & \\
\hline
& & & & \\
$\bf{W}=\Gamma_{0,2,0,\dots, 0}$& \scalebox{.6}{\ydiagram{2,2}} &$(0,2,0,\ldots,0)$ & $n(2n+3)$& Real \\
&  & & & \\
\hline 
\end{tabular}
\end{center}

\caption{
Some representations of USp($2n$) identified by their highest weight (in a base of fundamental weights), 
their Young tableaux, and their dimensions.   
The last row indicates if the representation is complex, real, or  pseudo-real.
\label{Table:RepCn}}
\end{table}

The group USp($2n$) is a maximal subgroup of SU($2n$) for all $n\geq 2$. 
The group USp($4$) is a special case as it is a maximal subgroup of both SU($4$) and SU($5$). 
In general, the group USp($2n$) can also be embedded in SU($2n+1$) as a non-maximal subgroup sitting in a sequence of two 
 maximal embeddings. 
Thus, we have the following two branching rules for $n\geq 2$ \cite{EK.KV}:
\begin{align}
 A_{2n-1} \downarrow C_n &:& \quad 
 \mathbf{adj}\ \textup{SU(2k)}\   \    \    \     &= \mathbf{adj}\ \textup{USp(2k)}\oplus\mathbf{\Lambda}_0^2, \\
  A_{2n}\downarrow A_{2n-1} \downarrow C_n&:&\quad \mathbf{adj}\ \textup{SU(2k+1)} &= \mathbf{adj}\ \textup{USp(2k)}\oplus\mathbf{\Lambda}^2_0\oplus \mathbf{V}\oplus \mathbf{V}\oplus \mathbf{1}.
\end{align}
For $n=2$, we also have in addition the maximal embedding 
\begin{equation}
 A_{4}\downarrow C_2:\quad  \mathbf{adj}\ \textup{SU(5)}  =\mathbf{adj}\ \textup{USp(4)} \oplus \mathbf{W}, 
\quad \mathbf{W}=\Gamma_{0,2}\   .
\end{equation}

The matter content of a USp($2n$)-model consists of the defining representation $\mathbf{V}$, the adjoint representation $\mathbf{Adj}=\text{Sym}^2_0 \textbf{V}$, and the traceless antisymmetric representation $\mathbf{\Lambda}^2_0$:
\begin{equation}
\mathbf{R}=  \mathbf{Adj}\oplus \mathbf{V}\oplus \mathbf{\Lambda}_0^2. 
\end{equation}
In a generic USp($2n$)-model we expect to see the weights of the representations $\mathbf{V}$ and $\mathbf{\Lambda}^2_0$  carried by rational curves components of singular fibers localized at different points. 
The sequence 
$A_{2n}\downarrow A_{2n-1} \downarrow C_n$ produces both $\mathbf{V}$ and $\mathbf{\Lambda}^2_0$ and therefore cannot correspond to localized matter.
The embedding  $ A_{2n-1}\downarrow C_n $ can describes the localized matter in the  representation  $\mathbf{\Lambda}_0^2$.

The  defining representation $\textbf{V}$ of C$_n$ is also its shortest non-trivial representation; the representation $\textbf{V}$ is minuscule,  pseudo-real, and of dimension $2n$. 
The traceless symmetric square  $\text{Sym}^2_0 \textbf{V}$  is irreducible, of dimension $n(2n+1)$, and corresponds to the adjoint representation. 
The second antisymmetric tensor product of $\textbf{V}$ is denoted $\boldsymbol{\Lambda}^2$ and  is not irreducible, but rather decomposes into a direct sum of the trivial representation and an irreducible representation of dimension 
 $(n-1)(2n+1)$ that we denote by $\boldsymbol{\Lambda}^2_0$. The trivial representation  is given by the trace of the antisymmetric tensor (note that the trace is defined with respect to the antisymmetric symplectic metric) and the 
representation $\boldsymbol{\Lambda}^2_0$ is the traceless antisymmetric projection of  $\boldsymbol{\Lambda}^2$. 
The representation $\boldsymbol{\Lambda}^2_0$ is irreducible,  quasi-minuscule, and real.
The representation  $\Gamma_{0,2,0,\cdots, 0}$ is irreducible of dimension $n(2n+3)$---see Table \ref{Table:RepCn}.

A Young tableau  is a finite left-justified array of boxes with weakly decreasing row
lengths. A Young tableau  is uniquely identified by  the list of the  lengths of its  rows.
 All the symmetric tensor products $\textbf{V}^{(k)}=\text{Sym}^k \textbf{V}$ are irreducible, but the exterior tensor products $\boldsymbol{\Lambda}^k$ are not. 
However, the kernel of the contraction map $\boldsymbol{\Lambda}^k \to \boldsymbol{\Lambda}^{k-2}$ induced by the invariant two-form of USp($2n$) is also an irreducible representation that we denote by $\boldsymbol{\Lambda}^k_0$. 
Irreducible representations of USp($2n$) are  in one-to-one correspondence with traceless tensors satisfying the symmetries of 
a Young tableau with no more than $n$ rows. 
The irreducible  representations of USp($2n$) 
with highest weight $\boxed{a_1\   \cdots\   a_n}$ 
in the basis of fundamental weights  
 is denoted by  $\Gamma_{a_1,a_2, \cdots, a_n}$. 
 We have the  restriction that all $a_i\geq 0$. 
 The irreducible representation  $\Gamma_{a_1,a_2, \cdots, a_n}$
  corresponds to the Young  tableau with  rows of  lengths 
$(\ell_1, \ldots, \ell_n)$ such that 
$\ell_k=\sum_{i=k}^n a_i$ for $k=1, \ldots, n$.

\end{document}